\documentclass[a4paper,USenglish,cleveref,autoref,thm-restate]{lipics-v2021}
\hideLIPIcs  

\usepackage[dvipsnames]{xcolor}
\usepackage[most]{tcolorbox}
\usepackage{lipsum}  
\usepackage{enumitem}
\usepackage[utf8]{inputenc}
\usepackage{tcolorbox}
\usepackage{amssymb,amsmath,amsthm}
\usepackage{xspace}
\usepackage{todonotes}
\usepackage{colortbl}

\usepackage{tcolorbox}
\tcbuselibrary{skins} 
\usepackage{xcolor}

\usepackage{mathtools}

\usepackage[T1]{fontenc}
\definecolor{anti-flashwhite}{rgb}{0.95, 0.95, 0.96}

\usepackage{float}
\usepackage{amsmath,amsfonts,amssymb,amsthm}
\usepackage{amsthm}
\usepackage{graphicx,color}
\usepackage{boxedminipage}
\usepackage{framed}
\usepackage{thmtools}
\usepackage{xspace}
\usepackage{todonotes}
\usepackage{footnote}

\usepackage{todonotes}
\usepackage{footnote}
 \usepackage{xifthen}
\usepackage{tabularx}
\usepackage{tikz} 
\usetikzlibrary{calc}

\newcommand{\cl}{{\sf cl}}

\newlength{\RoundedBoxWidth}
\newsavebox{\GrayRoundedBox}
\newenvironment{GrayBox}[1]%
{\setlength{\RoundedBoxWidth}{.93\textwidth}
	\def\boxheading{#1}
	\begin{lrbox}{\GrayRoundedBox}
		\begin{minipage}{\RoundedBoxWidth}}%
		{   \end{minipage}
	\end{lrbox}
	\begin{center}
		\begin{tikzpicture}%
			\node(Text)[draw=black!20,fill=white,rounded corners,%
			inner sep=2ex,text width=\RoundedBoxWidth]%
			{\usebox{\GrayRoundedBox}};
			\coordinate(x) at (current bounding box.north west);
			\node [draw=white,rectangle,inner sep=3pt,anchor=north west,fill=white] 
			at ($(x)+(6pt,.75em)$) {\boxheading};
		\end{tikzpicture}
\end{center}}

\newenvironment{defproblemx}[2][]{\noindent\ignorespaces%
	\FrameSep=6pt%
	\parindent=0pt%
	\vspace*{-1.5em}
	\ifthenelse{\isempty{#1}}{%
		\begin{GrayBox}{\textsc{#2}}%
		}{%
			\begin{GrayBox}{\textsc{#2} parameterized by~{#1}}%
			}
			\begin{tabular*}{\textwidth}{@{\hspace{.1em}} >{} p{1cm} p{0.8\textwidth} @{}}%
			}{
			\end{tabular*}%
		\end{GrayBox}%
		\ignorespacesafterend
	}

\newcommand{\pname}{\textsc}
\newcommand{\ProblemFormat}[1]{\pname{#1}}
\newcommand{\ProblemIndex}[1]{\index{problem!\ProblemFormat{#1}}}
\newcommand{\ProblemName}[1]{\ProblemFormat{#1}\ProblemIndex{#1}{}\xspace}

\newcommand{\probmwcut}{\ProblemName{Independent Multiway Cut}}

\definecolor{tyred}{rgb}{0.8, 0.0, 0.0}
\colorlet{mix}{red!50!black}

\providecommand{\keywords}[1]
{
	\textbf{\textit{Keywords---}} #1
}
\usepackage{hyperref}
\hypersetup{
	pdfnewwindow=true,      
	colorlinks=true,       
	linkcolor=mix,          
	citecolor=blue,        
	urlcolor=blue          
}

\newcommand{\I}{\mathcal{I}}
\newcommand{\M}{\mathcal{M}}

\newcommand{\stc}{$(s, t)${\text-cut}\xspace}

\newcommand{\sr}{\mathtt{s}\xspace}
\newcommand{\tr}{\mathtt{t}\xspace}
\newcommand{\br}{\mathtt{b}\xspace}

\newcommand{\yes}{\textup{\textsc{yes}}\xspace}
\newcommand{\no}{\textup{\textsc{no}}\xspace}

\newcommand{\fpt}{{\sf FPT}\xspace}

\newcommand{\red}[1]{{\color{red}#1}}
\newcommand{\lr}[1]{\left(#1\right)}
\DeclarePairedDelimiter{\LRb}{\lbrace}{\rbrace}
\newcommand{\LR}{\LRb*}
\newcommand{\cO}{\mathcal{O}}
\newcommand{\cT}{\mathcal{T}}
\crefname{step}{step}{steps}
\crefname{claim}{claim}{claims}
\newcommand{\cC}{\mathcal{C}}
\newcommand{\cN}{\mathcal{N}}

\newcommand{\mwcut}{{\tt IMWCut}\xspace}

\newcommand{\cM}{\mathcal{M}}
\newcommand{\cF}{\mathcal{F}}
\newcommand{\cI}{\mathcal{I}}

\newcommand{\Oh}{\mathcal{O}}

\newtheorem{obs}{Observation}

\newcommand{\rank}{{\sf rank}}
\newcommand{\indstcut}{{\sc Independent Vertex $(s, t)$-Cut}\xspace}
\newcommand{\repset}[1]{\subseteq^{#1}_{\text{rep}}}
\newcommand{\reach}{\mathsf{Reach}}
\newcommand{\cP}{\mathcal{P}}
\newcommand{\identity}{\textsf{identify}\xspace}
\newcommand{\cA}{\mathcal{A}}
\newcommand{\gindmcut}{{\sc GIMC}\xspace}
\newcommand{\gimc}{\gindmcut}
\DeclarePairedDelimiter{\tup}{\langle}{\rangle}
\newcommand{\typea}{type 1\xspace}
\newcommand{\typeb}{type 2\xspace}
\newcommand{\cJ}{\mathcal{J}}

\newcommand{\R}{r}
\newcommand{\cE}{\mathcal{E}}
\newcommand{\cS}{\mathcal{S}}
\newcommand{\col}{{\sf col}}
\newcommand{\betak}{k}

\newcommand{\ivstc}{IV$st$C\xspace}

\usepackage{algorithm}
\usepackage{algorithmicx}
\usepackage{algpseudocode}
\usepackage{mdframed}

\algtext*{EndWhile}
\algtext*{EndIf}
\algtext*{EndFor}

\title{Cuts in Graphs with Matroid Constraints}
\titlerunning{Cuts in Graphs with Matroid Constraints}

\author{Aritra Banik}{National Institute of Science, Education and Research, An OCC of Homi Bhabha National Institute, Bhubaneswar 752050,
Odisha, India}{aritra@niser.ac.in}{}{}

\author{Fedor V. Fomin}{University of Bergen, Norway}{Fedor.Fomin@uib.no}{https://orcid.org/0000-0003-1955-4612}{}

\author{Petr A. Golovach}{University of Bergen, Norway}{Petr.Golovach@ii.uib.no}{https://orcid.org/0000-0002-2619-2990}{}

\author{Tanmay Inamdar}
{Indian Institute of Technology Jodhpur, Jodhpur, India}{taninamdar@gmail.com}{https://orcid.org/0000-0002-0184-5932}{}

\author{Satyabrata Jana}{University of Warwick, UK}{satyamtma@gmail.com}{https://orcid.org/0000-0002-7046-0091}{}

\author{Saket Saurabh}{The Institute of Mathematical Sciences, HBNI, Chennai, India  \and University of Bergen, Norway }{saket@imsc.res.in}{https://orcid.org/0000-0001-7847-6402}{}

\Copyright{Aritra Banik, Fedor V. Fomin, Petr A. Golovach, Tanmay Inamdar, Satyabrata Jana, Saket Saurabh}

\funding{The research leading to these results has received funding from the Research Council of Norway via the project  BWCA (grant no. 314528).}

\authorrunning{A. Banik, F. V. Fomin, P. A. Golovach, T. Inamdar, S. Jana, S. Saurabh}


\nolinenumbers 

\EventEditors{John Q. Open and Joan R. Access}
\EventNoEds{2}
\EventLongTitle{42nd Conference on Very Important Topics (CVIT 2016)}
\EventShortTitle{CVIT 2016}
\EventAcronym{CVIT}
\EventYear{2016}
\EventDate{December 24--27, 2016}
\EventLocation{Little Whinging, United Kingdom}
\EventLogo{}
\SeriesVolume{42}
\ArticleNo{23}

\begin{document}

\maketitle

\begin{abstract}
{\sc Vertex $(s, t)$-Cut} and {\sc Vertex Multiway Cut} are two fundamental graph separation problems in algorithmic graph theory. We study matroidal generalizations of these problems, where in addition to the usual input, we are given a representation $R \in \mathbb{F}^{r \times n}$ of a linear matroid $\cM = (V(G), \cI)$ of rank $r$  in the input,  and the goal is to determine whether there exists a vertex subset $S \subseteq V(G)$ that has the required cut properties, as well as is independent in the matroid $\cM$. We refer to these problems as {\sc Independent Vertex \stc}, and {\sc Independent Multiway Cut}, respectively. We show that these problems are fixed-parameter tractable (\fpt) when parameterized by the solution size (which can be assumed to be equal to the rank of the matroid $\cM$). These results are obtained by exploiting the recent technique of flow augmentation [Kim et al.~STOC '22], combined with a dynamic programming algorithm on flow-paths \'a la [Feige and Mahdian,~STOC '06] that maintains a representative family of solutions w.r.t.~the given matroid [Marx, TCS '06; Fomin et al., JACM]. As a corollary, we also obtain \fpt algorithms for the independent version of {\sc Odd Cycle Transversal}. Further, our results can be generalized to other variants of the problems, e.g., weighted versions, or edge-deletion versions.
\end{abstract}

\newpage
    
\section{Introduction} \label{Sec:intro}
Studying problems in graphs and sets with additional constraints given by algebraic structures such as the matroid or submodularity is one of the well-established research directions. These versions not only generalize classical problems, but also show combinatorial interplay between graph structure and linear algebra. Some of the well-known problems in this direction include subsmodular versions of {\sc Vertex Cover}, {\sc Shortest Path}, {\sc Matching}, or {\sc Min-($s,t$)-Cut} \cite{GoelKTW09,jegelka2009notes} as well as 
coverage problems such as  {\sc Max Coverage} with matroid constraints~\cite{CalinescuCPV11,wolsey1982analysis}. In this paper we study some of the classical cut problems in the realm of parameterized complexity in the presence of constraints given by matroids. 

It is customary to begin the discussion on our framework with {\sc Vertex Cover} (given a graph $G$ and a positive integer $k$, find a subset of size at most $k$ that covers all edges), which is a poster-child problem in parameterized complexity. Among its myriad of generalizations studied in the literature, one of particular relevance to the present discussion is ranked version of \textsc{Vertex Cover} arising in the works of Lov\'asz \cite{lovasz1977flats,Lovasz19}. He defined \textsc{Rank Vertex Cover} in his study of critical graphs having $k$-sized vertex cover, by defining the notion of a \emph{framework}, also known as \emph{pregeometric graphs}, which refers to a pair $(G, \M)$, where $G$ is a graph and $\M = (V(G), \mathcal{I})$ is a matroid on the vertex set of $G$. In {\sc Rank Vertex Cover}, we are given a framework $(G, \cM)$ and a parameter $k$, and the goal is to determine whether $G$ has a vertex cover $S$ whose rank---which is defined as the maximum-size independent subset of $S$---is upper bounded by $k$. An \fpt algorithm for {\sc Rank Vertex Cover} (or indeed {\sc Rank $d$-Hitting Set}) parameterized by the rank is folklore. Meesum et al.~\cite{MeesumPSZ19} used {\sc Rank Vertex Cover} as a natural target problem for giving a compression from a different parameterization for \textsc{Vertex Cover}. However, even for a slightly more involved problem, such as \textsc{Feedback Vertex Set} (FVS), the ranked version immediately becomes \textsf{W}[1]-hard parameterized by the rank. This is similar to the phenomena seen when we consider weighted problems such as {\sc Weighted Vertex Cover} parameterized by weight~\cite{NiedermeierR03,ShachnaiZ17}.

Having faced with such a obstacle, we study an alternate generalization of FVS in frameworks. In \textsc{Independent FVS} (IFVS), where we are given a framework $(G, \cM)$, we want to determine whether there exists some $S \subseteq V(G)$ such that (i)  $S$ is an independent set in the matroid, and (ii) $G-S$ is acyclic. This happens to be a  ``sweet spot'' between the vanilla version and the ranked version -- we can adopt the textbook algorithms, and obtain a polynomial kernel and $\Oh^*(c^k)$ time \fpt algorithm when parameterized by $k$ ($\Oh^*(\cdot)$ suppresses polynomial factors in the input size.), which is equal to the solution-size (and can also be assumed to be equal to the rank of the matroid, by appropriately truncating it). This was the starting point of our result. Note that IFVS reduces to vanilla FVS when the matroid $\cM$ is a uniform matroid of rank $k$, and hence is \textsf{NP}-hard. Indeed, this gives us a recipe for defining the ``independent'' versions of classical vertex-subset problems in graphs.  We give a formal definition of a prototypical problem below.

\begin{tcolorbox}[enhanced,title={\color{black} {\sc Independent $\Pi$} {\footnotesize(where $\Pi$ is a vertex-subset problem)}}, colback=white, boxrule=0.4pt,
	attach boxed title to top center={xshift=-2cm, yshift*=-2.5mm},
	boxed title style={size=small,frame hidden,colback=white}]
	\textbf{Input:} \hspace*{5mm} A graph $G = (V,E)$, a matroid $ \M= (V,\I) $ of rank $r$. \\
	\textbf{Question:} Does there exist a set  $S \subseteq V(G)$ such that 
 \begin{itemize}[leftmargin=2cm]
    \item $S$ is an independent set in $\cM$, i.e., $S \in \cI$, and
     \item $S$ is a feasible solution for $\Pi$
 \end{itemize}
\end{tcolorbox}

On the other hand, even for a classically polynomial-time solvable problem such as \textsc{Vertex \stc}, the corresponding independent version can be shown to be \textsf{NP}-hard even when $\cM$ is a partition matroid \cite{Huang2018,ManurangsiSS23}. 

\subsection{Our Results and Methods.}
Let $G$ be a graph and $s,t$ be  a distinct pair of  vertices in $G$.  A vertex $(s,t)$-cut  is a set $S$ of vertices  such that $G-S$ has no path from $s$ to $t$. Given a graph $G$, a distinct pair of vertices $s,t \in V(G)$, and an integer $k$ , the \stc problem asks if a vertex $(s,t)$-cut of size at most $k$ exists.
Note that {\sc (Vertex) \stc} forms the bedrock of classical graph theory and algorithms via Menger's theorem and cut-flow duality. Thus, numerous constrained versions of {\sc \stc} have been studied in the literature, and matroid independence is a fundamental constraint that remains unexplored to the best of our knowledge. Another motivation for studying {\sc Independent Vertex \stc} (\ivstc) comes from the fact that it appears as a natural subproblem while trying to solve the independent version of {\sc Odd Cycle Transversal} (aka OCT -- given $G, k$, check whether we can delete at most $k$ vertices from $G$ such that the remaining graph is bipartite) via the so-called iterative compression technique~\cite{ReedSV04}. Thus, we would like to design an \fpt algorithm for \ivstc, and in turn IOCT, with running times that are as close to the running respective vanilla versions (note that \stc is in \textsf{P}, whereas OCT admits an $\Oh^*(3^k)$ time algorithm, see \cite{ReedSV04}).

There are two possible approaches that one could take for designing an \fpt algorithm for a constrained \stc problem, such as \ivstc. First approach is to reduce the treewidth of the input graph by using the methodology of Marx, O'Sullivan, and Razgon~\cite{MarxOR10,MarxOR13}, and then perform a dynamic programming on a bounded treewidth graph. While designing a dynamic programming algorithm for \ivstc on bounded-treewidth graphs appears to be straightforward using the technique of representative sets (\cite{Marx09,FominLPS16}), it is not obvious that we could use this result as a black-box--we have to be careful while performing any modifications to the graph, since we also have a matroid on the vertex-set. Further, since the treewidth bound obtained via this result is exponential in the size of the solution we are trying to find, it seems that such an approach would result in a double-exponential running time. Thus, we do not pursue this approach further, and the following question remains.

\begin{mdframed}[backgroundcolor=gray!10,topline=false,bottomline=false,leftline=false,rightline=false] 
 \centering
 \textbf{Question 1.} Does  {\sc Independent Vertex \stc} (and in turn, {\sc Independent Odd Cycle Transversal}) admit a single-exponential, preferably $2^{{\Oh}(k)} \cdot n^{\Oh(1)}$ \fpt algorithm? 
\end{mdframed}
Another promising approach to tackle this question is given by Feige and Mahdian~\cite{FeigeM06Separator}, who give a very general recipe for solving constrained minimum \stc problems. This is done via performing dynamic programming on the flow-paths obtained via Menger's theorem. Naturally, in order to use this approach, it is crucial that we are looking for a \emph{minimum} \stc with certain additional constraints. However, this appears to be a roadblock since a \emph{minimum independent} \stc can be much larger than the minimum (not-necessarily-independent) \stc, and hence Menger's theorem is not applicable. This is where the recent technique of flow augmentation~\cite{kim2021flow,DBLP:conf/stoc/0002KPW22} comes to the rescue. At a high level, this a powerful tool that can ``augment'' the graph by adding additional edges, thus lifting the size of minimum \stc to a higher value (hence, \emph{flow augmentation}), such that an unknown \emph{minimal} \stc becomes a \emph{minimum} \stc. Note that minimality of a solution can be assumed w.l.o.g. since it is compatible with independence as well as cut property. Indeed, all of the aforementioned problems, namely, {\sc FVS, OCT, \stc} satisfy the minimality condition. Hence, using flow augmentation, at the overhead of $2^{\Oh(k \log k)} \cdot n^{\Oh(1)}$, we can assume that we are looking for an independent set that is a minimum \stc in the augmented graph. However, due to certain technical obstacles while applying flow augmentation, we need to solve the \emph{directed version} of \stc using the approach of \cite{FeigeM06Separator}. Now, we perform a dynamic programming on the flow-paths given by a decomposition theorem from \cite{FeigeM06Separator} (which needs to be extended to directed graphs), and maintain a representative family of all partial solutions in the DP table. Note that for this, we assume that the matroid $\cM$ is given as a representation matrix $R \in \mathbb{F}^{r \times n}$, and the field operations can be performed in polynomial-time \footnote{A matroid $\cM$ is representable over a field $\mathbb{F}$ if there exists a matrix $R \in \mathbb{F}^{r \times n}$, and a bijection between the ground set of $\cM$ and the columns of $R$, such that a set is independent in $\cM$ iff the corresponding set of columns is independent in $\mathbb{F}$.}. Thus, in fact, our algorithm solves a more general problem where it outputs the representative family of all minimum independent $(s, t)$-cuts. This property is crucial when we want to use this problem as a subroutine to solve a more general problem, to which we return in the next paragraph. 
Note that for the sake of conceptual simplicity and to avoid references to success probabilities, we only give the deterministic running times arising from $\Oh^*(2^{\Oh(k^4 \log k)})$-time deterministic flow augmentation in the table and in the paper. However, note that by instead using the randomized version, we can improve the running time to $2^{\Oh(k \log k)}$, which takes us closer to our initial goal. Furthermore, due to versatility of the representative sets framework, our algorithms also extend to the weighted versions of these problems, where we want to find a minimum-weight independent set that is a feasible solution to $\Pi$. Note that one of the original motivations for the flow augmentation technique was solving one such problem, i.e., {\sc Bi-Objective Directed} (or weighted) {\sc \stc} \cite{DBLP:conf/stoc/0002KPW22}. Next, having (almost) answered \textbf{Question 1}, we look toward an even more general problem.

\begin{table}[t]
 \begin{center}
		\begingroup
		\setlength{\tabcolsep}{7pt} 
		\renewcommand{\arraystretch}{1.5} 
		\begin{tabular}{ |c | c | c |}
			\hline
   \rowcolor{anti-flashwhite}
   
			   \textbf{$\Pi$ in {\sc Independent $\Pi$}} &  \textbf{Matroid access}  & \textbf{Result}   \\\hline\hline
	\textsc{Feedback Vertex Set}& Oracle  &  $\Oh^*(10^{k})$ and poly kernel \\\hline
         \textsc{Vertex/Edge $(s, t)$-cut}& Representation & $\Oh^*(2^{\Oh(k^4 \log k)})$  \\\hline
         \textsc{Vertex $(s, t)$-Cut}& Oracle &  No $f(k) \cdot n^{o(k)}$ algorithm \\\hline
         \textsc{Odd Cycle Transversal} & Representation & $\Oh^*(2^{\Oh(k^4 \log k)})$  \\\hline
         \textsc{Multiway Cut} & Representation & $\Oh^*(2^{\Oh(k^4 \log k)})$  \\\hline
\end{tabular}
	\vspace{5mm}
	\caption{Our Results.\label{tab:ourresults}}
		\endgroup
  \end{center}
  
	\end{table}

\subparagraph{Independent Multiway Cut.} Let $G$ be a graph and $T \subseteq V(G)$ be a set of \emph{terminals}. A vertex multiway
cut is a set $S$ of vertices  such that every component of $G-S$  contains at most one vertex of $T$. Given a graph $G$, terminals $T$, and an integer $k$ , the \textsc{(Vertex) Multiway Cut} (MWC) 
problem asks if a vertex multiway cut of size at most $k$ exists.
MWC is a natural generalization of \stc.  This been studied extensively in the domain of parameterized and approximation algorithms. Note that the vanilla version of the problem was classically shown to be \fpt using the technique of important separators \cite{Marx06impseps}. On the other hand, it is not clear whether this technique is any useful in presence of matroid constraints, i.e., for \textsc{Independent Multiway Cut} (IMWC). Even when armed with an \fpt algorithm for {\sc Independent \stc} as a subroutine, it is not trivial to design an \fpt algorithm for IMWC. 

 The algorithm initially begins by trying to find a multiway cut (separator) $S$ of  size at most $k$ for the terminal set $T$ in $G$ -- here we completely ignore independence constraints, and such a separator can be found in $\Oh^*(2^k)$ using the result of Cygan et al.~\cite{CyganGH17}. Note that if the graph lacks such a separator, it also cannot have an independent separator. W.l.o.g. suppose that we find a separator $S$ that is also minimal. Hence, for every vertex $v$ in the set $S$, there exists a path that connects a pair of terminal vertices and intersecting the set $S$ solely at $v$. This fact brings us to the conclusion that every independent separator either includes $v$ or some vertices of terminal components (we say that a connected component in $G-S$ is a \emph{terminal component} if it contains some terminal). Moreover, each vertex of $S$ has a neighbor in at least two  terminal components. With this information, we attempt to design a recursive algorithm by assuming that the ``true solution'' $O$ to IMC of size $k$ either intersects $S$, or some terminal component (that can be identified). However, it is also possible that a terminal component contains the  entire solution $O$. In this case, we cannot make a recursive call to the same algorithm, since the parameter does not decrease. To address this scenario, we introduce the concept of a \emph{strong separator} (\Cref{def:strongseparator}), which is a minimal separator $S$ with the following properties: any minimal solution $O$ satisfies (i) $O \cap S \neq \emptyset$, or (ii) there exists a pair of terminal components---that can be identified in polynomial time---such that $O$ intersects at least one of them, with the size of intersection being between $1$ and $k-1$. Further, such a strong separator can be found in time $\Oh^*(2^k)$ time (\Cref{lem:strong}) by applying ``pushing arguments''. This lets us design a recursive divide-and-conquer algorithm, and due to property (ii), we know that in each recursive call, the parameter always strictly decreases. For handling the matroid independence constraints, in fact we need to strengthen the algorithm to return a representative family of all possible multiway cuts. This is where it is useful that our algorithm for \ivstc also returns such a family, since this algorithm is used in the base case of the algorithm (when $|T| = 2$). This strategy results $\Oh^*(2^{\Oh(k^4 \log k)})$ time deterministic (or $\Oh^*(2^{\Oh(k \log k)})$ time randomized) algorithm for IMWC, where the running time is dominated by the flow augmentation step required in the base case of \ivstc. The  \Cref{tab:ourresults} depicts all the results in this paper.

\subparagraph{Motivation for independence constraints.} One motivation for studying such independent versions of classical graph problems comes from the fact that matroid-independence capture several interesting constraints on the solution vertices. One prominent example of this is the \emph{colorful} versions of the problem that arise naturally from the \emph{color-coding} technique introduced by Alon, Yuster and Zwick~\cite{AlonYZ95}, where a problem such as {\sc $k$-Path}, i.e., determining whether the given graph contains a path of length at least $k$. While {\sc $k$-Path} itself is difficult to approach using combinatorial techniques, there is a simple dynamic programming algorihtm, if we want to find a colorful $k$-path, i.e., where we are given a graph whose vertices are colored from $\LR{1, 2, \ldots, k}$, and we want to determine whether the graph contains a path of length $k$ whose vertices are colored with distinct colors \cite{AlonYZ95}. Then, a standard argument shows that a hypothetical unknown solution will get colored with distinct colors if each vertex is assigned a color uniformly at random. On the flip side, colorful versions of the problems such as \textsc{Independent Set, Clique}, are also convenient starting points in various \textsf{W}[1]-hardness reductions. Note that colorful constraints is one of the simplest examples of matroid constraints -- it is simply a partition matroid, where the coloring of vertices corresponds to a partition of the vertex set, and the solution may contain (at most) one vertex from each partition. Slightly more general partition matroid constriants can capture ``local budgets'' on the solution from each class of vertices, which can be used to model different fairness constraints, i.e., we cannot delete too many vertices of the same type to achieve certain graph property, which can be ``unfair'' for that type. Such constraints on solution are rather popular in clustering \cite{HajiaghayiKK12,KrishnaswamyLS18}, which are indeed motivated from such fairness considerations.

\section{Preliminaries} {\bf Notations.} Let $[n]$ be the set of integers $\{1,\ldots, n\}$. For a (directed/undirected) graph $G$, we denote the set of vertices of $G$ by $V(G)$ and the set of edges by $E(G)$. For an undirected graph, we denote the set of non-edges by $\overline{E}(G)$, that is, $\overline{E}(G)= \binom{V(G)}{2} \setminus E(G)$, where $\binom{V(G)}{2}$ is the set of all unordered pairs of distinct vertices in $V(G)$; whereas for a directed graph $G$,$\overline{E}(G) \coloneqq (V(G) \times V(G)) \setminus E(G)$. For notational convenience, even in the undirected setting, we use the notation $(a, b)$ instead of $\LR{a, b}$ -- note that due to this convention, $(a, b) = (b, a)$ in undirected graphs.  For a un/directed graph $G=(V,E)$ and a subset $S$ of un/directed edges in $\overline{E}(G)$, we use the notation $G+S$ to mean the graph $G'=(V,E\cup S)$. In an un/directed graph $G$, and $A, B \in V(G)$, we say that $B$ is \emph{reachable} from $A$ if there is an un/directed path from some $u \in A$ to $v \in B$ in $G$. Note that reachability is symmetric in undirected graphs. For a vertex-subset $X \subseteq V(G) \setminus (A \cup B)$ (resp.~edge-subset $X \subseteq E(G)$), we say that $X$ is an $(A, B)$-cut if $B$ is not reachable from $A$ in $G-X$. Further, we say that $X$ is an \emph{inclusion-minimal $(A, B)$-cut} (or simply, a \emph{minimal $(A, B)$-cut}), if no strict subset of $X$ is also an $(A, B)$-cut. More generally, let ${\cal A} = \LR{A_1, A_2, \ldots, A_t}$ be a collection of vertex-subsets. We say that $X$ is a \emph{multiway cut} for ${\cal A}$ if $X$ is an $(A_i, A_j)$-cut for each pair $A_i, A_j \in {\cal A}$, and a minimal multiway cut is defined analogously. For vertex multiway cut $X \subseteq V(G)$ for ${\cal A}$ (resp.~$(A, B)$-cut), we also require that $X \cap A_i = \emptyset$ for all $A_i \in {\cal A}$ (resp.~$X \cap A = X \cap B = \emptyset$). Finally, when $A = \LR{s}, B = \LR{t}$, we may slightly abuse the notation and use \emph{$(s, t)$-cut} or \emph{$s$-$t$ cut} instead of $(\LR{s}, \LR{t})$-cut. When we have a graph $G$ and a subset of vertices $X \subseteq V(G)$, the operation of {\em identifying} the vertices of $X$ is defined as follows: we first eliminate the vertex set $X$, then introduce a new vertex $v_X$, make $v_X$ adjacent to all vertices in the open neighborhood of $X$, and finally remove any multi-edges.

\paragraph*{Matroids.} 

	
	
Let $\cM = (U, \cI)$ be a matroid. 
Each set $S \in \cI$ is called an independent set and an inclusion-wise maximal independent set is known as a basis of the matroid $\cM$. From the third property of matroids, it is easy to observe that every inclusion-wise maximal independent set has the same size; this size is referred to as the \emph{rank} of $\cM$, and denoted by $\rank(\cM)$. In the following, we will use $r$ to denote $\rank(\cM)$.

The rank of a set $P \subseteq U$ is defined as the size of the largest subset of $P$ that is independent, and is denoted by $\rank(P)$. Note that $\rank(P) \le \rank(Q)$ if $P \subseteq Q$. For any $P \subseteq U$, define the closure of $P$, $\cl(P) \coloneqq \LR{x \in U: \rank(P \cup \LR{x}) = \rank(P)}$. Note that $\cl(P) \subseteq \cl(Q)$ if $P \subseteq Q$. 

 Finally, we will use $U(\cM)$ and $\cI(\cM)$ to refer to the ground set and the independent-set collection of $\cM$, respectively.

A matroid $\cM = (U, \cI)$ is said to be \emph{representable} over a field $\mathbb{F}$ if there exists a matrix $R \in \mathbb{F}^{r \times n}$ where $r= \rank(\cM), |U|=n$, and a bijection from $U$ to the columns of $R$, such that for any $S \subseteq U$, $S \in \cI$ iff the corresponding set of columns is linearly independent over $\mathbb{F}$.

\begin{sloppypar}
    For a matroid $\cM = (U, \cI)$ and an element $u \in U$, the matroid obtained by contracting $u$ is represented by $\cM' = \cM/u$, where $\cM' = (U \setminus \LR{u}, \cI')$, where $\cI'= \LR{ S \subseteq U \setminus \LR{u} : S \cup \LR{u} \in \cI }$. Note that $\rank(\cM') = \rank(\cM) - 1$. For any $0 \le k \le r$, one can define a truncated version of the matroid as follows: $\cM_{k} = (U, \cI_{k'})$, where $\cI_{k} = \LR{S \in \cI: |S| \le k}$. Note that $\rank(\cM_k) = k$, and one can also think of $\cM_k$ as an intersection \footnote{Intersection of two matroids $\cM_1 = (U, \cI_1)$ and $\cM_2 = (U, \cI_2)$ defined over common ground set $U$ is the set family $(U, \cI)$, where $\cI = \LR{X \subseteq U: X \in \cI_1 \cap \cI_2}$. Note that $(U, \cI)$ is not necessarily a matroid in general.} of $\cM$ with a uniform matroid of rank $k$. It is easy to verify that the contraction as well as truncation operation results in a matroid. Furthermore, given a linear representation of a matroid $\cM$, a linear representation of the matroid resulting from contraction/truncation can be computed in (randomized) polynomial time \cite{DBLP:conf/acid/Marx06,DBLP:conf/icalp/LokshtanovMPS15}. 
\end{sloppypar}

We call a family of subsets (of $U$) of size $p \ge 0$ as a \emph{$p$-family}. Next, we state the following crucial definition of representative families. 

\begin{definition}[\cite{fomin2014efficient,DBLP:books/sp/CyganFKLMPPS15}] \label{repsetdefn}
	Let $\cM = (U, \cI)$ be a matroid. For two subsets $A, B \subseteq U$, we say that $A$ \emph{fits} $B$ (or equivalently $B$ \emph{fits} $A$), if $A \cap B = \emptyset$ and $A \cup B \in \cI$. 
    
    Let $\mathcal{A} \subseteq \cI$ be a $p$-family for some $0 \le p \le \rank(\cM)$. A subfamily $\mathcal{A}' \subseteq \mathcal{A}$ is said to \emph{$q$-represent} $\mathcal{A}$ if for every set $B$ of size
	$q$ such that there is an $A \in \cA$ that fits $B$, then there is an $A' \in \mathcal{A}'$  that also fits $B$. If  $\mathcal{A}'$ $q$-represents $\mathcal{A}$, we write $\mathcal{A}' \subseteq^q_{rep} \mathcal{A}$.
\end{definition}

The following proposition lets us compute represnetative families efficiently.

\begin{proposition}[\cite{fomin2014efficient,DBLP:books/sp/CyganFKLMPPS15}] \label{prop:repset} There is an algorithm that, given a matrix $R$ over a field $\mathbb{F}$, representing a matroid $\mathcal{M} = (U, \cI)$ of rank $k$, a $p$-family ${\cal A}$ of independent sets in $\mathcal{M}$, and an integer $q$ such that $p + q = k$, computes a $q$-representative family ${\cal A}'\subseteq ^q_{rep} \mathcal{A}$ of size at most
	$\binom{p+q}{p}$ using at most $\Oh(|{\cal A}|(\binom{p+q}{p})p^{\omega}+(\binom{p+q}{p})^{\omega-1})$
	operations over $\mathbb{F}$. Furthermore, an analogous result holds even if $p + q \le k$, albeit depending on the field $\mathbb{F}$, the algorithm may be randomized. \footnote{Here, we need to compute the representation of an appropriately truncated matroid. For this, only randomized polynomial algorithm is known in general \cite{Marx09}; however for certain kinds of fields, it can be done in polynomial time \cite{LokshtanovMPS18}.}
\end{proposition}

\begin{sloppypar}
    For a $p$-family ${\cal P}$ and a $q$-family ${\cal Q}$, we define ${\cal P} \bullet {\cal Q} \coloneqq \LR{P \cup Q : P \in {\cal P}, Q \in {\cal Q}, P \text{ fits } Q }$ (this operation is also known as subset convolution). The following properties of representative families follow straightfordly from the definitions, and formal proofs can be found in \cite{CyganFKLMPPS15}.
\end{sloppypar}

\begin{proposition}[\cite{CyganFKLMPPS15}]\label{obs:repsetprops}
    Let $\cM = (U, \cI)$ be a matroid and ${\cal A}$ be a $p$-family. 
    \begin{itemize}
        \item If ${\cal A}_1 \subseteq^{r-p}_{\text{rep}} {\cal A}$, and ${\cal A}_2 \subseteq^{r-p}_{\text{rep}} {\cal A}_1$, then, ${\cal A}_2 \subseteq^{r-p}_{\text{rep}} {\cal A}$.
        \item If ${\cal A}_1 \subseteq^{r-p}_{\text{rep}} {\cal A}$, then ${\cal A} = \emptyset$ iff ${\cal A}_1 = \emptyset$.
        \item Let ${\cal P}$ be a $p$-family and ${\cal Q}$ be a $q$-family. Let ${\cal P}' \repset{r-p} {\cal P}$, and ${\cal Q}' \repset{r-q} {\cal Q}$. Then, ${\cal P}' \bullet {\cal Q}'   \repset{r-p-q} {\cal P} \bullet {\cal Q}$.
    \end{itemize}
\end{proposition}

\section{\fpt Algorithm for \indstcut}\label{sec:mcut}

In this section, we design an \fpt algorithm for \indstcut. 
To this end, we define some notation. Consider a graph $G = (V, E)$, a matroid $\cM = (U, \cI)$ of rank $r$ where $V(G) \setminus \LR{s, t} \subseteq U$, and two vertices $s, t \in V(G)$.
By iteratively checking values $\betak = 1, 2, \ldots, $ and running the algorithm described below, we may inductively assume that we are working with a value $0 \le \betak \le \R$, such that there is no independent \stc  in $G$ for any $\betak' < \betak$. This follows from the fact that, our generalized problem, as defined below, requires the algorithm to return a set corresponding to $\betak'$ whose emptiness is equivalent to concluding that the there is no independent \stc of size $\betak'$.
It is important to note that, depending on whether we are addressing \indstcut as a standalone problem or as a sub-procedure, we may either obtain the smallest value of $k$ for free, or we may need to perform iterations. Regardless, we can safely assume that we are operating with a value of this nature.
 For any $0 \le k \le r$, let $$\cF(s, t, k) \coloneqq \LR{ S \subseteq V(G) \setminus \LR{s, t}: S \in \cI, |S| = k, S \text{ is a minimal } (s, t)\text{-cut.} }$$
In fact, we design an \fpt algorithm for a following generalization of the problem, stated below.

\begin{tcolorbox}[enhanced,title={\color{black} {\sc Generalized }\indstcut}, colback=white, boxrule=0.4pt,
	attach boxed title to top center={xshift=-2.5cm, yshift*=-2.5mm},
	boxed title style={size=small,frame hidden,colback=white}]
	
	\textbf{Input:} A graph $G$, a vertex set $Q \subseteq  V(G)$ containing  $s,t$, a representation $R \in \mathbb{F}^{r \times n}$ \hspace*{1.3cm} of a  matroid  $ \M= (U,\I)$ of rank $r$  where $V \setminus Q \subseteq U$,  an integer $0 \le k \le r$.  \\
	\textbf{Output:} $\cF'(s, t, k) \repset{r-k} \cF(s, t, k)$.
\end{tcolorbox}
Note that this generalizes \indstcut, since the latter problem is equivalent to determining whether $\cF(s, t, k)$ (or, equivalently, $\cF'(s, t, k) \repset{r-k} \cF(s, t, k)$, via \Cref{obs:repsetprops}) is non-empty for some $0 \le k \le r$. 

\begin{figure}[H]
	\centering
	\includegraphics[scale=0.7]{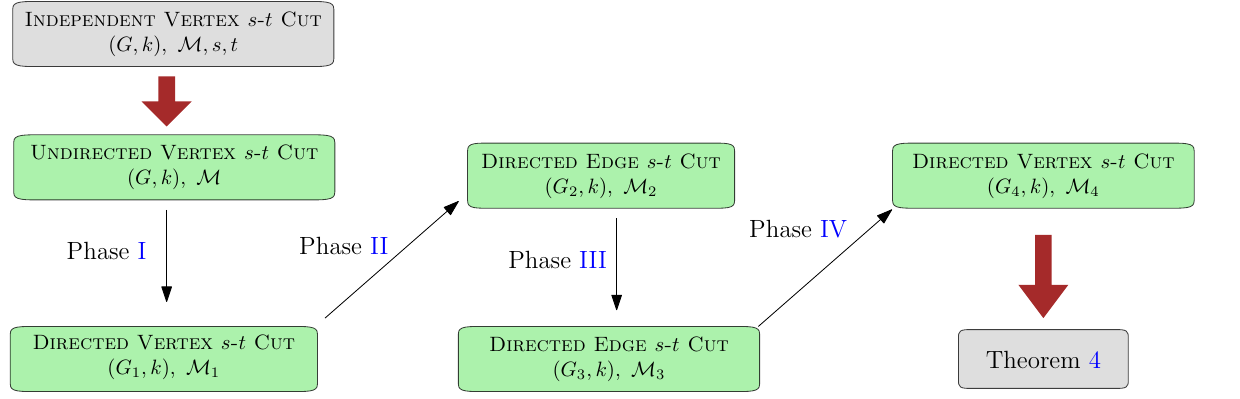}
	\caption{Schematic depiction of sequence of reductions from undirected vertex \stc to directed vertex \stc via flow augmentation. These details are omitted from the short version.}
	\label{fig:flowcut}
\end{figure}

\begin{restatable}{theorem}{stcuttheorem} \label{thm:stcuttheorem}
    Consider an instance $(G, \cM, s, t, Q, \betak)$ of {\sc Generalized \indstcut}, where there exists no independent \stc in $G$ of size less than $\betak$. Then,
	there exists a deterministic algorithm for {\sc Generalized} \indstcut that runs in time $2^{\Oh(k^4 \log k)} \cdot 2^{\Oh(\R)} \cdot n^{\Oh(1)}$, and outputs $\cF'(s, t, \betak) \repset{\R-\betak} \cF(s, t, \betak)$ such that $|\cF'(s, t, \betak)| \le 2^{\R}$. Here, $\cF(s, t, \betak)$ is the collection of all independent $(s, t)$-cuts of size $\betak$. 
\end{restatable}

\paragraph*{Flow augmentation.} Below we present the recent results regarding the so-called {\em flow-augmentation} technique \cite{DBLP:conf/soda/KimKPW21,DBLP:conf/stoc/0002KPW22}.
\begin{proposition}[Flow-augmentation (randomized) for Directed Graphs, \cite{kim2021flow,DBLP:conf/stoc/0002KPW22}]
	There exists a randomized polynomial-time algorithm that, given a directed graph $ G $, two vertices	$ s, t \in V (G) $, and an integer $ k $, outputs a set $ A \subseteq V (G) \times V (G) $ such that the following holds: for every minimal \stc $ Z \subseteq E(G) $ of size at most $ k $,  with probability $ 2^{-\cO(k \log k)}$ the edge set  $ Z $ remains an $ st $-cut in $ G + A $ and, furthermore,	$ Z $ is a minimum \stc in $ G + A $.
\end{proposition}

\begin{proposition}[Flow-augmentation (deterministic) for Directed Graphs \cite{kim2021flow,DBLP:conf/stoc/0002KPW22}]\label{prop:flowdir}
There exists an algorithm that, given a directed graph $ G $, two vertices $ s, t \in V (G) $, and an integer $k $, in time $2^{\cO(k^4 \log k)}  \cdot |V (G)|^{\cO(1)} $ outputs a set $ \mathcal{A} \subseteq 2^{V (G) \times V (G)} $ of size $2^{\cO(k^4 \log k)}  \cdot (\log n)^{\cO(k^3)} $ such that the following holds: for every minimal	\stc $ Z \subseteq E(G) $ of size at most $ k $, there exists $ A \in \mathcal{A} $ such that the edge set  $ Z $ remains an \stc in $ G + A $ and, furthermore,	$ Z $ is a minimum \stc in $ G + A $.
\end{proposition}

\subsection{Transforming the solution into a (directed) minimum cut via Flow Augmentation}


In this section we give a reduction from undirected vertex \stc to directed vertex \stc. This process consists of four phases. We adopt the following convention. Let $\betak$ denote the size of the minimum \stc and let $\R$ denote the rank of the matroid ${\cal M}$. If $\betak > \R$, then we immediately terminate and output \textsc{No}. Therefore, assume that $0 \le \betak \le \R$.

\subparagraph{Phase I. Undirected Vertex \stc $\Rightarrow$ Directed Vertex \stc.} We are given an instance $(G, s, t, \cM,Q, \betak)$ of {\sc Independent Undirected Vertex \stc} to an instance of $(G_1, s, t, \cM_1, Q_1, \betak)$ of {\sc Independent Directed Vertex \stc}.
First, given $G$, we define a directed graph $G_1$ as follows. 
  \begin{itemize}
      \item The vertex set remains the same, that is, $V(G_1)=V(G), Q_1= Q$.

      \item For each edge $e=uv$ in $G$, add two arcs $(u,v)$
and $(v,u)$ in $A(G_1)$.   
\end{itemize}

\noindent It is easy to see that $S$ is a vertex \stc in $G$ if and only if $S$ is a vertex \stc in $G_1$. 

\subparagraph{Matroid Transformation.} 
Here we construct a matroid $\cM_1$ for instance $(G_1,\betak)$. Initially, the matroid $\M$ is given as an input with $G$. We keep the same matroid for the new instance $(G_1,\betak)$, that is, $\M_1:=\cM$. Now, toward the correctness of the transformation we have the following claim. Let $\cF_1(s, t, \betak)$ denote the set of all minimal independent (w.r.t.$\cM_1$) $(s, t)$-cuts of size $k$ in $G_1$. The following claim is straightforward.

\begin{claim} \label{cl:bijection1}
    There is a natural bijection $f_{1}: \cF(s, t, \betak) \to \cF_1(s, t, \betak)$ that can be computed in polynomial-time. In particular, $\cF(s, t, \betak) \neq \emptyset \iff \cF_1(s, t, \betak) \neq \emptyset$.
  \end{claim}
\begin{claimproof}
   The proof is derived from the observation that the vertex set remains unchanged in both graphs $G$ and $G_1$. Therefore, any independent set in $\M$ will also be an independent set in $\M_1$, and vice versa (since $\M=\M_1$). 
\end{claimproof}

Let $\cF_1(s, t, \betak) \coloneqq \LR{X \subseteq V(G_1) \setminus \LR{s, t}: X \in \cI(\cM_1), |X| = \betak, X \text{ is a directed \stc}}$. It is straightforward to see that the task of finding $\cF'(s, t, \betak) \repset{q} \cF(s, t, \betak)$ is equivalent to finding $\cF'_1(s, t, \betak) \repset{q} \cF_1(s, t, \betak)$, using the bijection $f_1$ in \Cref{cl:bijection1}. Then, by using \Cref{lem:directedvtc}, we obtain the proof of \Cref{thm:stcuttheorem}.

\begin{figure}[ht]
\centering
\includegraphics[scale=0.5]{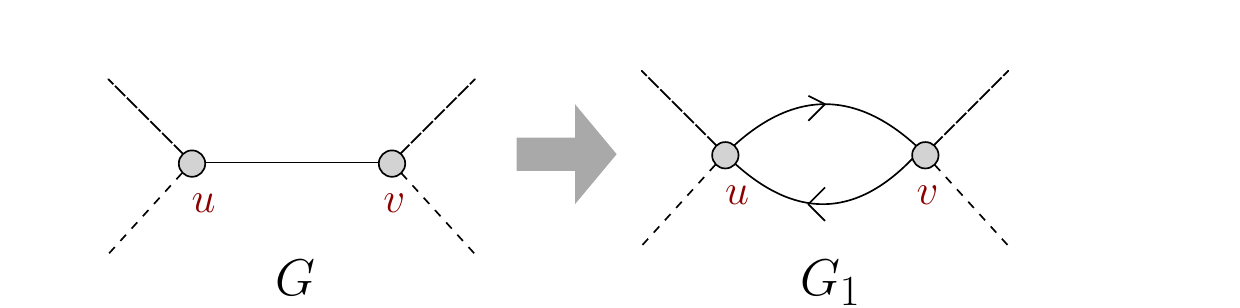}
	\caption{Illustration of Phase I.}
\label{fig:phasei}
\end{figure}

\subparagraph{Phase II. Directed Vertex \stc $\Rightarrow$ Directed Edge \stc.}
We are given an instance $(G_1, s, t, \cM_1, Q_1, \betak)$ of {\sc Generalized Independent Directed Vertex \stc}. We show how  to reduce it to an instance of $(G_2, s, t, \cM_2, Q_2, \betak)$ of {\sc Generalized Independent Directed Edge \stc}. Note that here $Q_2 \subseteq A(G_2)$. Given a directed graph $G_1$, we construct a directed graph $G_2$ as follows. 
  \begin{itemize}
      \item For each vertex $v$ in $G_1$, we introduce two vertices $v_{in}$ and $v_{out}$ in $G_2$. And add the arcs $(v_{in}, v_{out})$ to $A(G_2)$. We call them edges in $G_2$ corresponding to $V(G_1)$, or {\em v-edges}.

      \item For each arc $(u,v)$ in $A(G_1)$, we add the arc $(u_{out}, v_{in} )$ in $A(G_2)$. We call them edges in $G_2$ corresponding to $A(G_1)$, or {\em e-edges}

      \item The set $Q_2$ is the v-edges corresponding to the set $Q_1 \subseteq V(G)$. 
      \end{itemize}

       \begin{figure}[h]
	\centering
	\includegraphics[scale=0.5]{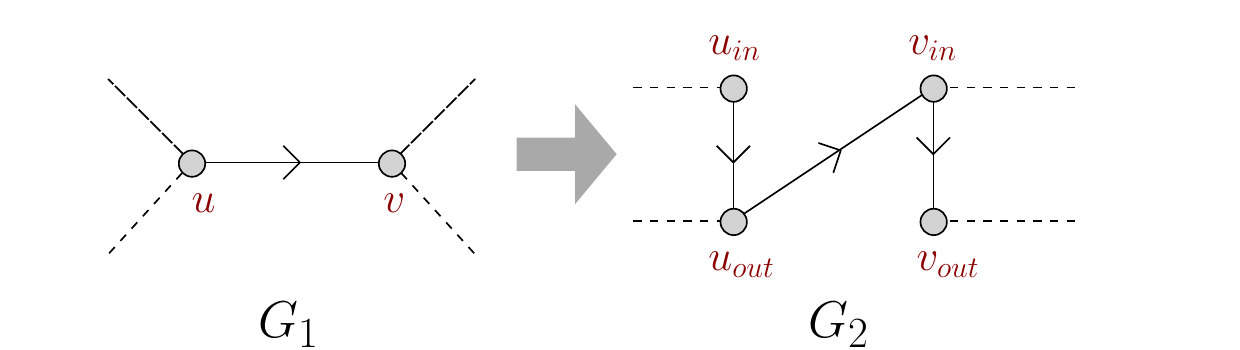}
	\caption{Illustration of Phase II.}
	\label{fig:phaseii}
\end{figure}

\subparagraph{Matroid Transformation.} 
 

  Here we construct a matroid $\cM_2$ for the instance $(G_2, s, t, \cM_2,\betak)$. Assume that the in the input we are given a graph $G_1$  along with a matroid $\M_1$. The ground set of $\cM_2$ is equal to $A(G_2)$. Any edge subset that contains a e-edge is deemed to be dependent. Further, using the natural bijection between a vetex $v$ in $V(G_1)$ and the corresponding v-edge in $V(G_2)$, we obtain a bijection between independent vertex-subsets in $\cM_1$, and make the corresponding v-edge-subsets independent in $\cM_2$. It is straightforward to obtain a representation of $\cM_2$, given a representation of $\cM_1$ -- we use the same vectors for v-edges, whereas other edges  is assigned a zero vector.

Let $\cF_1(s, t, \betak) \coloneqq \LR{S \subseteq A(G_1): S \text{ is a minimal edge \stc in $G_1$}, |S| = \betak, S \in \cI(\cM_1)}$, and $\cF_2(s, t, \betak) \coloneqq \LR{S \subseteq V(G_2): S \text{ is a minimal vertex \stc in $G_2$}, |S| = \betak, S \in \cI(\cM_2)}$. We prove the following claim.

\begin{claim}\label{claim:vst-est}
    There is a natural bijection $f_{12}: \cF_1(s, t, \betak) \to \cF_2(s, t, \betak)$ that can be computed in polynomial time. In particular, $\cF_1(s, t, \betak) \neq \emptyset \iff \cF_2(s, t, \betak) \neq \emptyset$.
\end{claim}
\begin{claimproof}
    In the forward direction, assume that $S \subseteq V(G_1)$ is a solution of size at most $\betak$.  We show that the arcs in $G_2$ corresponding to the vertices in $S \subseteq G_1$, more specifically, the set of arcs $\{(v_{in},v_{out}) \colon v \in S\}$ (denoted by $A_S$) of size at most $k$ is a solution of size at most $k$. On the contrary, assume that there is a path $P$ from $s$ to $t$ in $G_2-A_S$. Now, every arc on the path $P$, of the form $(v_{in}, v_{out})$ (Type I) for some $v \in V(G_1)$ or of the form $(u_{out}, v_{in})$ (Type II) for some $u,v \in V(G_1)$. Now, the existence of arcs of Type I implies that the vertex $v$ is not part of $S$. In the other case, the existence of Type II arcs implies that there is an arc $(u,v) \in A(G_1)$. Note that any pair of consecutive arcs on the path $P$ is of different type. With this we can find a path from $s$ to $t$ in $G_1-S$.     In the backward direction, assume that $Z \subseteq A(G_2)$ is a solution for $G_2$. Now we find a set of vertex $V_Z$ in $G_1$ in the following way: If there is an arc $(v_{in}, v_{out}) \in Z$ (Type I) for some $v \in V(G_1)$, we add $v$ to $V_Z$, if there is an arc of the form $(u_{out}, v_{in}) \in Z$ (Type II) for some $u,v \in V(G_1)$, we add $u$ in $V_Z$.  We show that $G_1-V_{Z}$ does not have a path from $s$ to $t$. On the contrary, assume that there is a path $P$ from $s$ to $t$ in $G_1-V_{Z}$.  Now, corresponding to the path $P$ in $G_1-V_{Z}$, we can find a path from $s$ to $t$ in $G_2 - Z$ in the following way. For every arc $(u,v)$ in $P$ we have an arc $(u_{out}, v_{in})$. Note that by our process, as the arc $(u,v)$ is present in $G_1-V_{Z}$, so both the arcs $(u_{in}, u_{out})$ and$(u_{out}, v_{in})$ are present in $G_2- Z$. So for every arc $(u,v) $ in $P$ we have a path from $u_{in}$ to $v_{out}$
in $G_2-Z$. With this we can find the path from $s$ to $t$ in $G_2-Z$. A contradiction.

We now demonstrate that the solution maintains the independence property when transitioning from $G_1$ to $G_2$, and similarly in the reverse direction. In the forward direction, assume that $S$ is an independent \stc of size at most $k$ in $G_1$. Notice that the arcs of $G_2$ corresponding to the vertices in $S$ (denoted by $A_S$) are indeed a \stc in $G_2$. Now, according to our construction, the vector set corresponding to $A_S$ in $\M_2$ is exactly the same as the vector set corresponding to $S$ in $\M_1$. Therefore, $A_S$ eventually becomes independent in $M_2$. In the reverse direction, assume that $Z$ is an independent \stc of size at most $k$ in $G_2$. Now $Z$ is independent in $\M_2$. This implies that $Z$ does not contain any edge of the form $(u_{out}, v_{in})$ for some $u,v \in V(G_1)$ (since no independent set of vectors contains a zero vector and the vector corresponding to the edge $(u_{out}, v_{in})$ is a zero vector in $\M_2$). So, every edge $Z$ is of the form $(v_{in}, v_{out})$ for some $v \in V(G_1)$. Notice that the vertices of $G_1$ corresponding to the arcs of $Z$ (denoted by $V_Z$) are indeed a \stc in $G_2$. Now, according to our construction, the vector set corresponding to $V_Z$ in $\M_1$ is exactly the same as the vector set corresponding to $Z$ in $\M_2$. Therefore, $V_Z$ eventually becomes independent in $\M_1$. Hence, the claim.     
\end{claimproof}

\begin{sloppypar}
\subparagraph{Algorithm.} Given an instance $(G_1, s, t, \cM_1, Q_1, \betak)$, we compute the instance $(G_2, s, t, \cM_2, Q_2, \betak)$, as mentioned above. From \Cref{claim:vst-est}, the task of finding $\cF'_1(s, t, \betak) \repset{\R-\betak} \cF_1(s, t, \betak)$ is reduced to $\cF'_2(s, t, \betak) \repset{\R-\betak} \cF_2(s, t, \betak)$, for which we use the algorithm from \Cref{lem:flowaugrepset}. Thus, we obtain the following lemma.
\end{sloppypar}

\begin{lemma} \label{lem:phase2repset}
	There exists a deterministic algorithm that takes an input $(G_1, s, t, \cM_1,Q_1, \betak)$ of {\sc Generalized Independent Directed Vertex \stc}, runs in time $2^{\Oh(\R)} \cdot 2^{\Oh(\betak^4)} \cdot n^{\Oh(1)}$, where $\R = \rank(\cM_1)$, and returns $\cF'_1(s, t, \betak) \repset{\R-\betak} \cF_1(s, t, \betak)$ of size at most $2^{\R}$. 
\end{lemma}

\subparagraph{Phase III. Flow augmentation}
\begin{definition}[Compatible augomenting set] \label{def:cas}
	Let $Z \in E(G)$ be a minimal \stc in $G$. We say that $A \in V(G) \times V(G)$ is a \emph{compatible augmenting edge set} for $Z$ in $G$ if (i) $Z$ remains an edge \stc in $G+A$, and (ii) $Z$ is a minimum edge \stc in $G+A$. 
	\\Let $\cA = \LR{A_1, A_2, \ldots}$ is a collection of subsets of $V(G) \times V(G)$. If for $Z \in E(G)$, there exists some $A_i \in \cA$ that is a compatible augmentating edge set for $Z$ in $G$, then we use $A_{\cal A}(Z)$ to denote such an $A_i$ -- if there are multiple such $A_i$'s, then we pick arbitrarily; whereas if there is no such $A_i \in \cA$, then we define $A_{\cal A}(Z) = \bot$. When $\cA$ is clear from the context, we will omit it from the subscript.
\end{definition}

Assume that we have an instance $G_2$ for directed edge \stc.   Using \cref{prop:flowdir}, in time $2^{\cO(\betak^4 \log \betak)}  \cdot |V (G)|^{\cO(1)} $ we will find a set $ \mathcal{A} \subseteq 2^{V (G) \times V (G)} $ of compatible augmenting edge-sets with $\alpha \coloneqq |\mathcal{A}|= 2^{\cO(k^4 \log \betak)}  \cdot (\log n)^{\cO(\betak^3)} $. 
Now, from $G_2$, we create $\alpha$ instances in the following way: for each $A \in \mathcal{A}$, we create an instance $G_2^A= G_2 + A$.

\begin{claim}\label{claim:flowdec}
   For any $Z \subseteq E(G_2)$ of size $\betak' \le \betak$, $Z$ is an edge \stc in $G_2$ if and only if there exists an $A_{\cal A} \in \mathcal{A}$ such that $ Z $ is a minimum edge \stc of size $\betak'$ in $ G_2^A$. 
\end{claim}
\begin{proof}
 The forward direction can be derived from the assertion made in \Cref{prop:flowdir}. In the reverse direction, since $G_2^A$ is created by adding some edges to $G_2$, it is evident that any edge \stc in $G_2^A$ also remains an edge \stc in $G_2$. 
\end{proof}

\subparagraph{Matroid Transformation.}  
Now, fix some $A \in \cA$ and we describe the matroid associated with the instance corresponding $G_2^A$, called $\cM^A$. The ground set $U(\cM^A_3)$ is equal to $U(\cM_2) \cup A = E(G_2^A)$. The collection of independent sets in $\cM^A$ is same as that in $\cM_2$ -- any set $S \subseteq E(G^A_2)$ with $S \cap A \neq \emptyset$ is deemed dependent in $\cM^A_2$. It is straightforward to verify that $\cM^A_2$ is a matroid. Given a representation for $\cM_2$, it is easy to come up with a representation for $\cM^A_2$ -- we simply associate each edge in $A$ with a zero vector. Let $\cF^A_2(s,t, \betak) = \LR{X \subseteq E(G^A_2) : X \text{ is an independent minimum edge \stc in $G^A_2$}, |X| = \betak}$. Now, from \Cref{claim:flowdec} and the above discussion, we obtain the following corollary.
\begin{observation} \label{obs:augmentation}
	For each $Z \in \cF_2(s, t, \betak)$, $Z \in \cF_2^A(s, t, \betak)$, where $A = A_{\cal A}(Z)$. 
\end{observation}

\begin{sloppypar}
	\subparagraph{Algorithm.} For each $A \in \cA$, we create an instance $(G_2^A, s, t, \cM^A, Q^A, \betak)$, and compute $\cF^{'A}_2(s, t, \betak) \repset{\R-\betak} \cF^A_2(s, t, \betak)$ using the algorithm of \Cref{lem:directedvtc}. Then, we compute $\tilde{\cF_2}(s, t, \betak) \coloneqq \bigcup_{A \in \cA} \cF^{'A}_2(s, t, \betak)$, and then we compute $\cF'_2(s, t, \betak) \repset{\R - \betak} \tilde{\cF_2}(s, t, \betak)$. We show the following. 
\end{sloppypar}

\begin{sloppypar}
    \begin{lemma} \label{lem:flowaugrepset}
	There exists a deterministic algorithm that takes input an instance $(G_2, s, t, \cM_2,Q_2, \betak)$ of {\sc Generalized Independent Directed Edge \stc}, runs in time $2^{\Oh(\R)} \cdot 2^{\Oh(\betak^4)} \cdot n^{\Oh(1)}$, where $\R = \rank(\cM_2)$, and returns $\cF'_2(s, t, \betak) \repset{\R-\betak} \cF_2(s, t, \betak)$ of size at most $2^{\R}$. 
\end{lemma}
\end{sloppypar}

\begin{proof}
	Consider any $Z \in \cF_2(s, t, \betak)$, i.e., $Z \subseteq E(G_2)$ is a minimal independent edge \stc of size $\betak$. From \Cref{obs:augmentation}, it follows that $Z \in \cF_2^A(s, t, \betak)$, where $A = A_{\cal A}(Z)$. Now consider any $Y \subseteq U(\cM_2) \subseteq U(\cM_3)$ such that $Y$ fits $Z$. Then, there exists some $Z' \in \cF^{'A}_2(s, t, \betak)$ such that $Y$ fits $Z'$. Clearly, $Z' \subseteq E(G_2)$, since $Z'$ is independent in $\cM^A_2$, and $Z'$ is of size $\betak$. Further, $Z'$ is a minimum edge \stc in $G^A_2$ of size $\betak$. Then, by \Cref{claim:flowdec}, it follows that $Z' \in \cF_2(s, t, \betak)$, which shows the lemma.
\end{proof}

\subparagraph{Phase IV. Directed Edge \stc $\Rightarrow$ Directed Vertex \stc.} We are given an instance $(G_3, s, t, \cM_3, Q_3, \betak)$ of {\sc Generalized Independent Directed Edge \stc}. We show how to reduce it to an instance of {\sc Generalized Independent Directed Vertex \stc}. Given a directed graph $G_3$, define a directed graph $G_4$, and set $Q_4 \subseteq V(G_4)$ as follows. 
  \begin{itemize}
      \item for each  vertex $ u $ in $V(G_3)$, we add  $ \betak+1 $ vertices  $ u_1, u_2, \ldots, u_{\betak+1}$.  We call them vertices in $G_4$ corresponding to $V(G_3)$, or \emph{v-vertices}.
      
      \item For each arc $ (u,v) \in A(G_3)$, we add a vertex $ e_{uv} $ in $V(G_4)$, and arcs $ (u_i, e_{uv}) $ and $ (e_{uv}, v_i) $  to $A(G_4)$ for every $i \in [\betak+1]$. We call them vertices in $G_4$ corresponding to $A(G_3)$, or \emph{e-vertices}.

      \item The set $Q_4$ is the e-vertices  corresponding to the set $Q_3 \subseteq A(G_3)$.
      \end{itemize}

       \begin{figure}[h]
	\centering
	\includegraphics[scale=0.5]{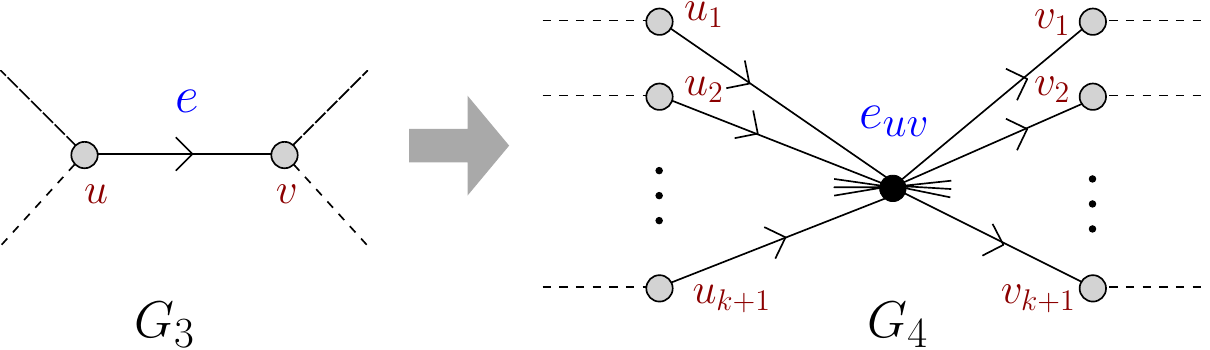}
	\caption{Illustration of Phase IV.}
	\label{fig:phaseiv}
\end{figure}

\subparagraph{Matroid Transformation.} 
Here we construct a matroid called $\cM_4$ for the instance $(G_4,s,t,\cM_4,\betak)$ using a similar approach to what we did in phase II. The ground set of $\cM_4$ is equal to $V(G_4)$. Since $s, t \in V(G_4)$ cannot be part of any solution, we prevent them from being part of any independent set. Any vertex subset that contains a v-vertex is deemed to be dependent. Further, using the natural bijection between an edge in $A(G_3)$ and the corresponding e-vertex in $V(G_4)$, we obtain a bijection between independent edge-subsets in $\cM_3$, and make the corresponding e-vertex-subsets independent in $\cM_4$. It is straightforward to obtain a representation of $\cM_4$, given a representation of $\cM_3$ -- we use the same vectors for e-vertices, whereas each v-vertex is assigned a zero vector.


Let $\cF_3(s, t, \betak) \coloneqq \LR{S \subseteq E(G_3): S \text{ is a minimal edge \stc in $G_3$}, |S| = \betak, S \in \cI(\cM_3)}$, and $\cF_4(s, t, \betak) \coloneqq \LR{S \subseteq V(G_4): S \text{ is a minimal edge \stc in $G_4$}, |S| = \betak, S \in \cI(\cM_4)}$. We prove the following claim.

\begin{claim}\label{claim:vst-est2}
    There is a natural bijection $f_{34}: \cF_3(s, t, \betak) \to \cF_4(s, t, \betak)$ that can be computed in polynomial time. In particular, $\cF_3(s, t, \betak) \neq \emptyset \iff \cF_4(s, t, \betak) \neq \emptyset$.
\end{claim}

\begin{claimproof}
       In the forward direction, assume that $Z \subseteq A(G_3)$ is a solution of size at most $\betak$.  We show that the vertices in $G_4$ corresponding to the arcs in $Z$ in $G_3$, more specifically, the set of vertices $\{e_{uv} \colon (u,v) \in Z\}$ (denoted by $V_Z$)  is a solution of size at most $\betak$ in $G_4$. On the contrary, assume that there is a path $P$ from $s$ to $t$ in $G_4-V_Z$. Now, every arc on the path $P$, either of the form $(u_i, e_{uv})$ or $(e_{uv}, v_j)$ for some $u,v \in V(G_3)$ and $i,j \in [\betak+1]$. Now if the vertex $e_{uv}$ is present in the path $P$, that means $(u,v) \notin Z$. Now we get a path from $s$ to $t$ in $G_3-Z$ using the arcs $\{e_{uv} \colon e_{uv} \in V(P)\}$, which is a contradiction.    In the backward direction, the proof is similar to the forward direction.Aassume that $S \subseteq V(G_4)$ is a solution for $G_4$. Let there be a vertex $u_i$ for some $i \in [\betak+1]$ such that $u_i \in S$. We claim that $S \setminus u_i$ remains a solution for $G_4$. Note that since there are $\betak+1$ many vertices (copy) for each vertex $u \in G_3$, we must have $\{u_j \colon j \in [\betak+1]\} \setminus Z \neq \emptyset$. Let $u_j \in \{u_j \colon j \in [\betak+1]\} \setminus Z$. On the contrary, if $S \setminus u_i$ is not a solution, then there is a path from $s$ to $t$ in $G_4 - (S \setminus u_i)$. And that path must pass through the vertex $u_i$. Now, from this we can find a path from $s$ to $t$ in $G_4 -S$ just bypassing the path replacing $u_i$ by $u_j$ (since both vertices have the same neighbor sets). A contradiction. So $S \setminus u_i$ remains a solution for $G_4$. Repeating this argument, we can obtain that the set of all vertices in $S$ that correspond to some arc in $G_3$ (denoted by $S'$) remains a solution for $G_4$.  Now we show that the arcs in $G_3$ corresponding to the vertices in $S'$ in $G_4$, more specifically, the set of arcs $\{(u,v)\colon e_{uv} \in S'\}$ (denoted by $A_{S'}$)  is a solution of size at most $k$ in $G_3$. 

    We now demonstrate that the solution maintains the independence property when transitioning from $G_1$ to $G_2$, and similarly in the reverse direction. Since the vectors for all the edges in $G_3$ remain the same for all the vertices in $G_4$ that correspond to those edges, any independent set in $\M_3$ will still be an independent set in $\M_4$. On the contrary, since no independent set in a matroid contains a zero vector, any independent set in $G_4$ cannot include the vertex in $G_4$ that corresponds to a vertex in $V(G_3)$. Therefore, any independent set in $\M_4$ is also an independent set in $\M_3$. This completes the proof.    
\end{claimproof}

\begin{sloppypar}
	\subparagraph{Algorithm.} We take the instance $(G_3, s, t, \cM_3, Q_3,\betak)$ of {\sc Generalized Independent Directed Edge \stc} and reduce it to an instance $(G_4, s, t, \cM_4, Q_4, \betak)$ of {\sc Generalized Independent Directed Vertex \stc}. Then, we use \Cref{theo:dstcut} and compute $\cF'_4(s, t, \betak) \repset{\R-\betak} \cF_4(s, t, \betak)$ in time $2^{\Oh(r)} \cdot n^{\Oh(1)}$, such that $|\cF'_4(s, t, \betak)|$. Then, by using the bijection-inverse $f_{34}^{-1}$ from \Cref{claim:vst-est2} we obtain $\cF'_3(s, t, \betak)$. We observe, due to the properties of the bijection $f_{34}$, it holds that $\cF'_3(s, t, \betak) \repset{\R-\betak} \cF_3(s, t, \betak)$. We summarize this in the following lemma.
\end{sloppypar}

\begin{sloppypar}
    \begin{lemma} \label{lem:directedvtc}
	There exists a deterministic algorithm that takes input an instance $(G_3, s,t, \cM_3, Q_3,  \betak)$, where $E(G_3) \subseteq U(\cM_3)$ and $r = \rank(\cM_3)$, and computes $\cF'_3(s, t, \betak) \repset{\R-\betak} \cF_3(s, t, \betak)$ of size at most $2^r$ in time $2^{\Oh(r)} \cdot n^{\Oh(1)}$. 
\end{lemma}
\end{sloppypar}

\subsection{Algorithm For Directed Vertex \stc}
  
Assume that we are given a directed graph $ G=(V,E) $, two specified vertices $ s $ and $ t $, a vertex set $Q \subseteq V(G)$ containing $s, t$,  a matroid $\cM = (U, \cI)$ such that $V \setminus Q \subseteq U$. In this section, path always refers to a directed path unless mentioned otherwise. We say that a vertex set $S \subseteq V(G) \setminus \LR{s, t}$ is a (directed) \stc in $G$ if $G-S$ does not contain any $s$ to $t$ path. Further, we say that an \stc $S \subseteq V(G) \setminus \LR{s, t}$ is an independent \stc if $S \in \cI$.

Similar to the previous sections, for $0 \le \betak \le \R$, we define $$\cF(s, t, \betak) \coloneqq \LR{X \subseteq V(G) \setminus \LR{s, t}: X \text{ is an independent \stc of size $k$}}.$$

Here is a formal description of the problem we consider. 
  
  	\begin{tcolorbox}[enhanced,title={\color{black} \sc{Generalized Independent Directed Vertex \stc} (GIDVC)}, colback=white, boxrule=0.4pt,
  	attach boxed title to top center={xshift=-.9cm, yshift*=-3mm},
  	boxed title style={size=small,frame hidden,colback=white}]
  	
  	\textbf{Input:}  A directed graph $G$, a linear matroid $ \M= (U,\I) $,  a  set $Q \subseteq V(G)$  containing \hspace*{1.3cm} $ s,t$,  such that $V(G) \setminus Q \subseteq U$, and two non-negative integers  $\betak, q$ such  that \hspace*{1.3cm} $\betak+q \le \R = \rank(\cM)$, and $k$ is the size of minimum vertex \stc. \\
  	\textbf{Output:}   Return $\cF'(s, t, \betak) \repset{q} \cF(s, t, \betak)$ of bounded size.
  \end{tcolorbox}

   In this section, we adopt the approach of Feige and Mahdian~\cite{FeigeM06Separator} to our end. Many of the definitions are directly borrowed from the same (albeit for directed graphs). 






\begin{definition}[Critical and non-critical vertex]
{\em 	A vertex $ v $ of $ G $ is called {\em critical} if every collection of $\betak $ vertex-disjoint paths from $s$ to $t$ contains $ v $. A vertex	is {\em non-critical} if it is not critical. }
	\end{definition}

We now have the following fact, which directly follows from the above definition.

\begin{proposition} \label{prop:criticalprop}
	A vertex $ v $ is critical if and only if there is a directed vertex \stc of size  $\betak$ containing $ v $.
\end{proposition}

\begin{definition}[Connecting pairs]
{\em For a pair of vertices $ u $ and $ v $ in $ G $, we say that $ u $ is connected to $ v $ in $ G $ if there is a directed path from  $ u $ to  $ v $  such that  every internal vertex of the path is non-critical.}
\end{definition}

 Fix a collection of $\betak$ vertex-disjoint paths $ P_1 , P_2 , \ldots, P_{\betak}$	from $ s $ to $ t $. By definition, each critical vertex must be on one of these paths. In addition, \cref{prop:criticalprop} implies the following.

 \begin{observation}\label{obs:criticalsoln}
 For every $S \in \cF(s, t, k)$, all $v \in S$ are critical vertices. 
 Furthermore, distinct vertices of $S$ belong to distinct paths from $\LR{P_1, P_2, \ldots, P_\betak}$. 
\end{observation}

 For each $ P_i $ , let $v_{i,1} , \ldots, v_{i,z_i}$ be the sequence of critical vertices of $ P_i $ in the order in which they appear on this path from $ s $ to $ t $. Here, $ z_i $ denotes the number of critical vertices in $ P_i $.  To simplify notation, we define $ v_{i,0} = s $ and $ v_{i,z_i +1} = t $ for every $ i $, and treat the vertices $ s $ and $ t $ as critical. Note that in the problem definition, we are not allowed to be chosen as solution vertices (i.e., cut vertices).  Let $ \Omega $ be the set of all $\betak$-tuples $ \boldsymbol{a} = (a_1, \ldots, a_{\betak}) $ where $ 1 \leq a_i \leq z_{i} $ for every $ i $. In other words, each $ \boldsymbol{a} \in \Omega $ corresponds to one way of selecting a critical vertex from each $ P_i $. Note that this does not have to correspond to an \stc in $ G $, since there are edges and vertices in $ G $ that are not in the paths $ P_i $’s.

\begin{definition}[Prefix subgraph]
{ \em	The prefix subgraph $  G[\boldsymbol{a}] $ defined by $ a \in \Omega$	is an induced subgraph of $ G $ with the vertex set defined as follows:
	
	\begin{itemize}
		\item[$ \bullet $] 	a critical vertex $ v_{ i,j} $ is in $  G[\boldsymbol{a}] $ if and only if $ j \leq a_i$ ;
	    \item[$ \bullet $] a non-critical
	    vertex $u $ is in $  G[\boldsymbol{a}] $ if and only if all critical vertices that are connected to $ u $ are in $  G[\boldsymbol{a}] $. 			
	\end{itemize}
  The last two layers of $  G[\boldsymbol{a}] $ are the set of critical vertices $ v_{ i,j} $ such that $ a_{i- 1} \leq j \leq a_i$ .}
\end{definition}

\begin{lemma}[Decomposition lemma, \cite{FeigeM06Separator}]\label{lem:decom}
	There is a sequence $ \boldsymbol{a}^1 , \ldots, \boldsymbol{a}^p \in \Omega$  such that
	\begin{enumerate}
		\item[(a)] $ \boldsymbol{a}^1 = (1, \ldots , 1) $ and $ \boldsymbol{a}^p = (z_1 , \ldots, z_{\betak}) $;
		
		\item[(b)] for every $ h = 2,\ldots, p $, $ \boldsymbol{a}^h - \boldsymbol{a}^{h-1}$ is a vector with exactly one entry equal to one, and zero elsewhere; and
		
		\item[(c)]  for every $ h = 1, \ldots , p $, every critical vertex of $  G[\boldsymbol{a}^h]$, except possibly the vertices in the last two layers of $  G[\boldsymbol{a}^h]$, is not connected to any critical vertex not in $  G[\boldsymbol{a}^h]$.
  	\end{enumerate}
	\end{lemma}

 \begin{proof}
	
	We construct the sequence inductively. It is clear that $ G[\boldsymbol{a}^1] $
satisfies condition (c) above. Assume that we have constructed the sequence up to $ \boldsymbol{a}^{h-1} $. We show that there is some $ \boldsymbol{a}^h $ that satisfies the conditions of the lemma. To this end, we show that there is an $i \in [k] $ such that $ v_{i, \boldsymbol{a}_i^{h-1}-1} $ is not connected to any critical vertex outside $ G[\boldsymbol{a}^{h-1}] $. Assume, for contradiction, that such an $ i $ does not exist. This means that for each $ i $, there is a critical vertex $ v_{j,\ell} $ outside $ G[\boldsymbol{a}^{h-1}] $ (that is, with $ \ell > \boldsymbol{a}_j^{h-1} $) such that $ v_{i, \boldsymbol{a}_i^{h-1}-1} $ is connected to $ v_{j,\ell} $.

Now, we construct an auxiliary directed graph $ H $ with vertex set $ [\betak] $. For each $ i, j \in  [\betak] $, there is a directed edge
	 from $ i $ to $ j $ in $ H $ if $ v_{i, \boldsymbol{a}_i^{h-1}-1} $
	 is connected to a critical vertex $ v_{j,\ell} $	 on $ P_j $ with $ \ell>  \boldsymbol{a}_j^{h-1}$. By our assumption, every vertex in $ H $ has outdegree at least one, and therefore $ H $ has a cycle. Consider the shortest cycle $ C = i_0, \ldots, i_{f-1}$ in $ H $. Each edge $ i_bi_{(b+1) \mod f}  $ of this cycle corresponds to a path $ Q_b $ from $ v_{i_b, \boldsymbol{a}_{i_b}^{h-1}-1} $ to $ v_{i_{b+1}, \ell_{b+1}} $
	 for some $ \ell_{b+1} > \boldsymbol{a}_{i_{b+1}}^{h-1} $, such that all internal vertices of this path are non-critical.  We show that the existence of this cycle contradicts the fact that $ v_{i_0, \boldsymbol{a}_{i_0}^{h-1}} $
	 is a critical vertex.	 As $ v_{i_0, \boldsymbol{a}_{i_0}^{h-1}} $ is a critical vertex, there exists a vertex $st$-cut $ S $ of size $\betak$ that contains the vertex $v_{i_0, \boldsymbol{a}_{i_0}^{h-1}}$. The removal of $ S $ divides the graph into several connected components. We call the vertices that can be reached from $ s $ on the graph $ G-S $ {\em silver} vertices, and the vertices from which we can reach $ t $ in $ G-S $ {\em tan} vertices. Since $ S $ contains exactly one vertex from each $ P_i $ and $ \ell_0 > \boldsymbol{a}_{i_0}^{h-1}$ the vertex $ v_{i_0, \ell_0} $ must be a tan vertex. This vertex is connected by the path $ Q_{f-1} $ from the vertex $ $ $v_{i_{f-1}, \boldsymbol{a}^{h-1}_{i_{f-1}}-1} $. Since all the vertices of $ Q_{f-1} $ are non-critical and hence not in $ S $, the vertex $v_{i_{f-1}, \boldsymbol{a}^{h-1}_{i_{f-1}}-1} $ is  either	 tan or in $ S $. Therefore, the vertex $ v_{i_{f-1}, \ell_{f-1}} $ must be tan.  Similarly, we can argue that the vertices $ v_{i_{f-2}, \ell_{f-2}}, \ldots, v_{i_{1}, \ell_{1}} $	 are all	 tan. However, $ v_{i_{1}, \ell_{1}}$ is connected by the path $ Q_0 $ from $ v_{i_0, \boldsymbol{a}_{i_0}^{h-1}-1} $, and	the latter vertex must be silver, since the only vertex on $ P_{i_0} $	 that is in $ S $ is $ v_{i_0, \boldsymbol{a}_{i_0}^{h-1}} $. This gives us the desired contradiction.
	  				 
	   The above argument shows that there is a vertex $ v_{i, \boldsymbol{a}_{i}^{h-1}-1} $	 that is not connected to any critical vertex outside $ G[\boldsymbol{a}^{h-1}] $. Such a vertex can be found efficiently by trying all possibilities. Now, we simply let $ \boldsymbol{a}_i^h=a_i^{h-1}+1 $ and $ \boldsymbol{a}_j^h=\boldsymbol{a}_j^{h-1} $ for every $ j \neq i $. It is easy to see that this choice of $ \boldsymbol{a}^h $ satisfies the conditions of the lemma.
\end{proof}

\begin{remark}
	The generating sequence in the decomposition lemma maintains the invariant that at any step, every critical vertex outside the current subgraph is not connected from any critical vertex other than the ones in the last two layers of the current subgraph. In other words, the last two layers of the current prefix subgraph act as the {\em interface} of this subgraph to the rest of $ G $.
	\end{remark}

 Next, we define the notion of a valid coloring that essentially captures partial solutions that will be used to perform dynamic programming.

\begin{definition}[Valid coloring]\label{def:coloring}
{\em 	A valid coloring of a prefix subgraph $  G[\boldsymbol{a}] $ is 	a partial coloring of the vertices of $  G[\boldsymbol{a}] $ with colors \textbf{s}ilver $ (\sr) $, \textbf{t}an $ (\tr) $, and	\textbf{b}lack $ (\br) $ such that		
	\begin{itemize}
		\item[$ \bullet $] (black vertices) for each $ i $, there is at most one vertex of $  G[\boldsymbol{a}] $ on $ P_i $ that is		colored black; furthermore, $ s $ and $ t $ cannot be colored black;
		
		\item[$ \bullet $] (silver and tan critical vertices) for each $ i $ and $ j \leq a_i $ , if there is no $ j' \leq a_i $ such that $ v_{i,j'}$ is
		colored black, then $ v_{i,j}$ must be colored silver; if there is such		a $ j' $, then $  v_{i,j} $ must be colored silver if $ j < j' $ and tan if $ j > j'$ ;
		
	\item[$ \bullet $] No silver-colored critical vertex is connected to a tan-colored critical vertex; and 
	
	\item[$ \bullet $] (coloring for non-critical vertices) every non-critical vertex that is connected from at least one silver critical vertex is colored silver; every non-critical vertex that is connected to at least one tan critical vertex is colored tan;

	\item[$ \bullet $] (some vertex remains uncolored) every	non-critical vertex that is neither connected from any silver critical vertex  nor connected to any tan	critical vertex remains uncolored.
\end{itemize}
We use $\col$ to  denote the set of all possible valid colorings of $G$ that color vertex $t$ tan. We use $\cS$ to denote all minimum vertex $(s, t)$-cuts in $G$.
}

\end{definition}

\begin{remark}
	At a high level, the first condition in \Cref{def:coloring} says that, each path contains one black vertex (must be critical by \cref{obs:criticalsoln}). The second condition says, in each path the colors of critical vertices follow the pattern either ``silver $ <\cdots < $ silver $<$ black $<$  tan $<\cdots< $ tan''. The third condition in the above definition guarantees	that no non-critical vertex is connected  from  a silver critical vertex as well as a tan	critical vertex. The fourth and fifth conditions specifies a well-defined coloring for non-critical vertices. Notice that due to third condition, among two exactly at most one situation can occur for any non-critical vertex.
\end{remark}

  \begin{figure}[h]
	\centering
	\includegraphics[scale=0.5]{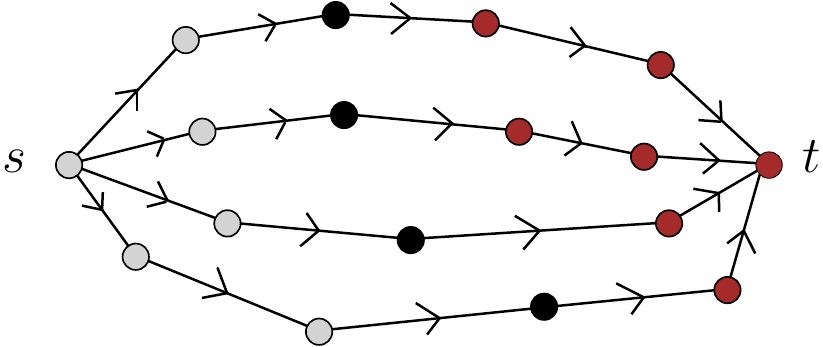}
	\caption{An illustration of  a valid coloring on critical vertices.}
	\label{fig:valid coloring}
\end{figure}

We start with following lemma.

\begin{lemma}\label{lemma:validcoloring}
    There is a bijection $f: \col \to \cS$, i.e. between the set of all valid colorings of $G$ that color vertex $t$ tan, and the set of all minimum vertex $(s, t)$-cuts in $G$.
\end{lemma}

\begin{proof}
    Let $\cS$ denote the set of all minimum vertex $(s,t)$-cuts in $G$ and $\col$ denote the set of all possible valid colorings of $G$ that color vertex $t$ tan. We show that there is a bijection $f \colon \col \rightarrow \cS$. 
    
    Consider a set $S \in \cS$. As $|S|=k$, this means that $S$ intersects exactly one vertex on every path in $\{P_i \colon i \in [k]\}$. Now we provide a coloring in the following way. We color $\sr$ for each vertex in $S$. For the remaining vertices, for each $i \in [k]$ if there is a vertex $v $ in the path $P_i$ such that $v$ is critical and $v$ lies on the subpath of $P_i$ from $s$ and the black vertex (which is the solution vertex), then we color $\sr$ to the critical vertex. Similarly,   for each $i \in [k]$ if there is a vertex $v $ on the path $P_i$ such that $v$ is critical and $v$ lies on the subpath of $P_i$ from $s$ and the black vertex (which is the solution vertex) and $t$ then we color $\tr$ to the critical vertex. We color the vertex $t$ to $\tr$. Now there is no non-critical vertex which is connected from a silver critical vertex as well as connected to a tan critical vertex, as in this case we get a $(s,t)$-path in $G-S$, a contradiction.  Now we can give a coloring to some non-critical vertices with the following way.  For  every non-critical vertex that is connected from at least one silver critical vertex we color that vertex $\sr$; every non-critical vertex that is connected to at least one tan critical vertex we color that vertex $\tr$. Now as per the \cref{def:coloring}, we have obtained a valid coloring. 

    Now consider a valid coloring $c$ with   $\tr$ color to the vertex $t$.  Let $S'$ be the set defined by all the black vertices in the coloring. Now we use the property of valid coloring. In the first property, we know that  each path contains at most one black vertex, that means $|S'| \leq k$. As the vertex $t$ obtained the color $\sr$ (as our consideration), and no silver-colored critical vertex is connected to a tan-colored critical vertex (third property), maintaining the the second property we have that each path contains exactly one black vertices, so $|S'|=k$. In a contrary assume that $S'$ is not a vertex $(s,t)$-cut then there is a path $P$ from $s$ to $t$ such that no vertex in $V(P)$ is colored black. So, every critical vertex on the path is colored either $\sr$ or $\tr$. As the vertex $s$ and $t$ gets the color $\sr$ and $\tr$, respectively, there must be a path from the $\sr$ colored critical vertex to the $\tr$ colored critical vertex, which is a contradiction, as it violates the third property of valid coloring. Hence, the proof.   
\end{proof}

\begin{theorem}\label{theo:dstcut}
 There exists a deterministic algorithm for {\sc Generalized Directed Independent Vertex \stc}  that runs in time $2^{\Oh(\R)}$ time, and returns $\cF'(s, t, k) \repset{r-k} \cF(s, t, k)$ of size at most $2^r$. Here, $\betak$ denotes the size of the minimum \stc  and $\R$ denotes the rank of the matroid ${\cal M}$.
\end{theorem}

\begin{proof}
Note that by \cref{lemma:validcoloring} any valid coloring of the graph $ G $ that colors the vertex $ t $ tan $(\tr)$  corresponds to a \stc in $ G $. Conversely, any rank \stc in $ G $ gives rise to a valid coloring of $ G $ that colors $ t $ tan. 
Recall that our goal is to efficiently compute $\cF'(s, t, k) \repset{\R-\betak} \cF(s, t, k)$ of bounded size. By \cref{lemma:validcoloring}, there is a bijection between set of all minimum vertex $(s, t)$-cuts in $G$ and valid colorings with $t$ is colored tan. Therefore, we design a dynamic programming-based  algorithm along the decomposition lemma (\cref{lem:decom}) to compute ``partial representative family'' that represents (partial) valid coloring. We describe this DP below.

 We use \cref{lem:decom} to construct the sequence $  \boldsymbol{a}^1, \ldots, \boldsymbol{a}^p$ in polynomial time. Based on this sequence, we define a table $A$ of dimension  $ p \times 3^{2\betak} \times \betak$ $ A $ as	follows.

\begin{itemize}
	\item[$ \bullet $] The entry $ A[h, y,i]$ is indexed an integer $ h \in [p] $, a string $ y \in \{\sr, \tr, \br\}^{2\betak} $, and an integer $i \in [\betak]$.
	
	 \item[$ \bullet $] Fix some $h \in [p], y \in \LR{\sr, \tr, \br}^{2\betak}, i \in [\betak]$. We say that a valid coloring $\chi: V(G[\boldsymbol{a}^h]) \to \LR{\sr, \tr, \br}$ {\em respects} the tuple $(h, y, i)$ if it colors the vertices in the last two layers of $G[\boldsymbol{a}^h]$ according to $y$, and $|B_{\chi}| = i$, where $B_\chi \coloneqq \chi^{-1}(\br)$. We define ${\cal S}(h, y, i)$ to be the set of all such subsets $B_\chi$ over all valid colorings $\chi$ of $G[\boldsymbol{a}^h]$ that respect $(h, y, i)$. In the entry $A(h, y, i)$, we will store $\widehat{\cal S} \repset{\R-i} {\cal S}(h, y, i)$. 


\end{itemize}

 The base entries $ A[1, \cdot, \cdot]$ can be computed by inspection. We now give a procedure to compute $ A[h,y,i]$, based on the entries $ A[h-1, \cdot, \cdot]$.  The last two layers of $ G[\boldsymbol{a}^h ] $ differ from the
last two layers of $ G[\boldsymbol{a}^{h-1} ] $ in that one vertex (say $ v $) was added and one vertex (say $ u $) was removed. Let $ w $ denote the vertex which is in the last layer of $ G[\boldsymbol{a}^{h-1} ] $ but in the second last layer of $ G[\boldsymbol{a}^{h} ]$. Note that the vertex $w$ always exists.   Technically, if $ u = s $, then $ s $ may still belong to the last two layers of $ G[a^h ] $, but the treatment of this case is only simpler than in the case $ u \neq s $ and is omitted. We try all three possible colors for $ u $.
For each color, we first check if in combination with $ y $ it violates the condition for a valid coloring. If `not, then we do the next step. We compute the color of all the vertices that are in $ G[\boldsymbol{a}^h ] $ but not in $ G[\boldsymbol{a}^{h-1} ] $ using the fourth condition in \cref{def:coloring}. Note that by condition (c) of \cref{lem:decom}, the color of any such vertex can be uniquely specified. We perform a case analysis based on $ y(v) $ and define a family $F$ that we will eventually store in $A[h, y, i]$.  
Let $ y' $ be a coloring of the vertices in the last two layers of $ G[\boldsymbol{a}^{h-1} ]$. 
We say that $y$  and $y'$ are \emph{compatible} with each other, if they 
agree on coloring of the vertices that are in common in their respective last two layers. 
We denote this by $y \propto y'$.

\begin{description}
	\item[(i) $ y(v)=\br $] As $y$ is a valid coloring we have $y(u)=y(w)=\sr$. In this case, if $ y' \propto y $ then we must have $ y'(u)= y'(w) = \sr$. Therefore, there is a unique coloring in the last two layers in the vertices of $ G[\boldsymbol{a}^{h-1} ] $ that is compatible with $ y $.  Let  $ F= \{Q \cup \{v\} \colon Q \in A[h-1,y',i-1],~y' \propto y\} $.  So here $ |F|=|A[h-1,y', i-1]| \le 2^{\R} $.
	
	\medskip
	
	\item[(ii) $ y(v)=\sr $] As $y$ is a valid coloring, we have $y(u)=y(w)=\sr$. In this case, if $ y' \propto y $ then we must have $ y'(u)= y'(w) = \sr$. Therefore, there is a unique coloring in the last two layers in the vertices of $ G[\boldsymbol{a}^{h-1} ] $ that is compatible with $ y $. Let $   F= A[h-1, y', i]$.  So, here $ |F|=|A[h-1,y',i]| \leq 2^{\R} $.
	
	\medskip 
	
	\item[(iii) $ y(v)= \tr $] As $y$ is a valid coloring, we have either $ y(u)= y(w) = \tr$, or $ y(u)= \br,~ y(w) = \tr $, or $ y(u)= \sr,~y(w) = \br $.  In this case, if $ y' \propto y $ then we must have either $ y'(u)= y'(w) = \tr$, or $ y'(u)= \br,~  y'(w) = \tr $, or $ y'(u)= \sr,~y'(w) = \br $. Here we distinguish these by two subcases based on  $y'(w)$. Recall that $ y(w)=y'(w) $. 
	
\begin{itemize}
	\item If $ y'(w)=\br $ then we must have $ y'(u)= \sr $. So in this case,  there is a unique coloring in  the last two layers in the vertices of $ G[\boldsymbol{a}^{h-1} ] $ which is compatible with $ y $.  Let  $   F= A[h-1, y', i]$.  So here $ |F|=|A[h-1,y',i]| \leq 2^{\R} $.
	
	\medskip 
 
	\item If $ y'(w)=\tr $ then $ u $ can take color either black or tan. So there are at most  two different colorings in  the last two layers in the vertices of $ G[\boldsymbol{a}^{h-1} ] $ which is compatible with $ y $. Let $ F'= \bigcup_{y' \propto y} A[h-1,y', i]$. So  here $ |F'|=|\bigcup_{y' \propto y} A[h-1,y', i]| \leq 2 \cdot 2^{\R} $. Now we use  the tool of {\em representative sets} (\cref{prop:repset}) to find a set $ F \repset{\R-i} F'$ of size at most $ 2^{\R} $. 
\end{itemize}

\end{description}

Finally, we will let $A[h, y, i] \gets F$.	Let ${\cal Y}$ denote the set of all strings that color the vertex $t$ tan in $G[\boldsymbol{a}^p]$. Note that $|{\cal Y}| \le 2^{\betak}$. 
Given the table $ A $, one can easily check the existence of a desired \stc by checking the entries $ A[p, y, \betak]$ for all strings $ y \in {\cal Y}$. Formally if $ A[p,y, \betak]\neq \emptyset$ for some string $ y$ where $ y(t)=t $ then we return \yes to our problem on the instance $ G $. For the more general problem, when we seek the representative family, we compute $\cF' \repset{\R-\betak} \bigcup_{y \in {\cal Y}} A[p, y, \betak]$ and return $\cF'$. 


\end{proof}



\stcuttheorem*

    

\section{(Odd) Cycle Hitting}\label{sec:cyclehitting}

In this section, we consider matroidal generalizations of two quintessential cycle hitting problems in parameterized complexity, namely \textsc{Feedback Vertex Set} (FVS), and \textsc{Odd Cycle Transversal} (OCT), respectively. 

\subsection{{\sc Independent FVS}} \label{subsec:fvs}

Here, we sketch an argument that gives a polynomial kernel for \textsc{Independent FVS} parameterized by the solution size $k$, which we can assume to be equal to the rank $\R$ of the matroid, \emph{even in the oracle access model}.
However, we assume that the memory used to store oracles does not contribute to the input size, and this is important for our kernelization result. 
The standard degree 0 and 1 rules are also applicable here. Note that in standard FVS, we can replace an internal vertex $v$ from an induced degree-2 path $P$ by a vertex $u$ of degree at least 3 that is at an end of $P$, since any cycle that passes through $v$, also passes through $u$. However, in the presence of matroid constraint, such an argument no longer works -- since replacing $v$ by $u$ in a solution may not result in an independent set. Nevertheless, we can bound the relevant vertices from a long induced path of degree 2 by $\R+1$, using the following proposition. 

\begin{restatable}{lemma}{oraclerepset}\label{lem:oraclerepset}
	Let $\cM$ be a matroid of rank $\R$ given via oracle access, and let $\mathcal{A}$ be a $1$-family of subsets of $U(\cM)$. Then, $\mathcal{A}' \subseteq^{\R-1}_{rep} \mathcal{A}$ can be computed in polynomial time, where $|\mathcal{A}'| \le \R$.
\end{restatable}
\begin{proof}
	Let $A = \LR{x : \LR{x} \in \mathcal{A}}$ be the set of underlying elements corresponding to the sets of $\mathcal{A}$. We compute an inclusion-wise maximal subset $R \subseteq A$ that is independent in $\cM$. Note that this subset can be computed using $O(|U|\R)$ queries, where $r = \rank(\cM)$. We claim that $\mathcal{A}' = \LR{\LR{y}: y \in R}$ is a $1$-representative set of $\mathcal{A}$.
	
	Consider any $u \in A$ and $X \subseteq U$ of size $\R$ such that $u \in X$ and  $X \in \cI$. We will show that there exists some $u' \notin R$ such that $X \setminus \LR{u} \cup \LR{u'} \in \cI$. Since $R$ is an inclusion-wise maximal independent subset of $A$, it follows that $u \in \cl(R) \subseteq \cl(R \cup (X \setminus \LR{u}))$. This means that $\rank(R \cup (X \setminus \LR{u})) = \rank(R \cup (X \setminus \LR{u}) \cup \LR{u}) = \rank(R \cup X)$. However, $\rank(R \cup X) \ge \rank(X)$. Therefore, we obtain that $\rank(R \cup (X \setminus \LR{u})) \ge \rank(X)$. This implies that there exists some $u' \in R$ such that $X \setminus \LR{u} \cup \LR{u'} \in \cI$. 
\end{proof}

\subparagraph{Polynomial kernel.} Consider any set $P$ that induces a maximal induced degree 2 path in $G$. We replace $P$ with a degree-2 path induced by $P'$, where $P'$ is obtained via \Cref{lem:oraclerepset}. After performing each such replacement, we obtain a graph $G'$, and let $V' \subseteq V(G')$ be the set of vertices of degree at least $3$. 
Now consider a graph $H'$ that is obtained by $G'$ by replacing each maximal induced degree-2 path in $G'$ by a single edge, and we call such an edge a \emph{special edge}. Note that $H'$ is exactly the graph we would obtain by using standard polynomial kernelization algorithm on $G$. By standard arguments, we know that the number of vertices and edges in $H'$ of degree at least $3$ is bounded by $\Oh(\R^2)$. Since each edge in $H'$ may correspond to a path of length $k$ induced by degree-2 vertices, it follows that the number of vertices and edges in $G'$ is bounded by $\Oh(\R^3)$.

\subparagraph{\fpt algorithm.} Brute-forcing on the $\Oh(\R^3)$ kernel immediately shows that {\sc Independent FVS} is \fpt parameterized w.r.t.~$\R$ with running time $\R^{\Oh(\R)} + n^{\Oh(1)}$ -- note that this also requires only an oracle access to the matroid. However, we can design a faster, single exponential \fpt algorithm for the problem via adapting a couple of approaches for standard FVS. 

We first use the kernelization algorithm and obtain an equivalent instance of size $\Oh(\R^3)$. In fact, instead of working with the graph $G'$, we work with the graph $H'$ where the number of vertices and edges is $\Oh(\R^2)$, where we denote the subset of special edges in $E(H)$ by $E_{sp}$. For simplicity, we use $G$ and $H$ instead of $G'$ and $H'$. Now we use the iterative compression technique, where we are working with an induced subgraph $H[X]$, for some $X \subseteq V(H)$. Note here that $X$ is a subset of vertices of degree at least $3$ in $G$. Further, let $Z \coloneqq V(G) \setminus V(H)$ denote the set of vertices of degree $2$ in $G$, and for $X \subseteq V(H)$, let $Z_X \subseteq Z$ denote the subset of vertices that appear on a path $P$ connecting two vertices $u, v \in X$. Note that such a path is represented by a special edge $uv \in E(H)$, and indeed the special edge $uv$ is present in $H[X]$. 

Inductively we assume that we have an FVS $S \subseteq X$ of size $\R+1$ in the graph $H[X]$, and we want to determine whether there exists a vertex-subset $O \subseteq X$, and $F \subseteq E(H[X]) \cap E_{sp}$ of special edges, such that: (i) $|O \cup F| \le \R$, and (ii) for each edge $e \in F$, there exists a degree-2 vertex $v_e$ in the corresponding degree-2 path, such that $O \cup \bigcup_{e \in F} v_e \in \cI$. 

We try to guess $Y \coloneqq O \cap S$, assuming such an $O$ indeed exists. In this step, in addition to deleting $Y$ from the graph, we additionally contract the matroid $\cM$ on the set $Y$ of size $p \le \R$, and obtain another matroid $\cM/Y$. Since $\cM$ is provided to us using its representation matrix, we can obtain that of $\cM_Y$ in polynomial-time (see, e.g., \cite{Marx09}). Thus, we can reduce the problem to {\sc Disjoint Independent FVS}, where we have an FVS $N \coloneqq S \setminus Y$ for the graph $H[W] = G-Y$, and we want to find an independent FVS of size $t \coloneqq \R-p$ in $G'$ that is disjoint from $N$. For this problem, we can emulate a standard branching algorithm, albeit with some modifications. In the course of the algorithm, we also maintain a set $\cE \subseteq E_S$ of special edges, which is initialized to an empty set. Let $W = X \setminus N$, and note that $H[W]$ is a forest.

\begin{algorithm}[htbp]
	\caption{\texttt{DisjointFVS}$(H = (X, E), \cM, W, N, \cE, t)$}
	\label{alg:disjointfvs}
	\begin{algorithmic}[1]
        \If{$t = 0$}
            \If{$H$ contains a cycle}
                \Return $\bot$
            \Else
                \State Find common base (if exists) for $\cM$ and $\cM_1$ using \cite{Cunningham86}, where
                \State $\cM_1 = (V(G), \cI_1)$ where $\cI_1 = \LR{S \subseteq V(G)|S \cap P_e| \le 1 : e \in \cE } $
                \State \Return common base $S$ if it exists, \textbf{else} $\bot$
            \EndIf
        \EndIf
		\If{$\exists v \in W$ with  neighbors $u_1, u_2 \in N$ in the same component of $H[N]$}
            \State \textbf{if } $u_1v \in E_s$ \textbf{ then } $S_1 \gets \texttt{DisjointFVS}(H, \cM, W, N, \cE \cup \LR{u_1v}, t-1)$ \textbf{else} $S_1 \gets \bot$
            \State \textbf{if } $u_2v\in E_s$ \textbf{ then } $S_2 \gets \texttt{DisjointFVS}(H, \cM, W, N, \cE \cup \LR{u_2v}, t-1)$ \textbf{else} $S_2 \gets \bot$
            \State $S_3 \gets \texttt{DisjointFVS}(H-v, \cM/v, W-v, N, t-1)$
            \State \textbf{if} $S_1 \neq \bot$ \textbf{ then } \Return{$S_1$}
            \State \textbf{else if} $S_2 \neq \bot$ \textbf{ then } \Return{$S_2$}
            \State \textbf{else if} $S_3 \neq \bot$ \textbf{ then } \Return{$S_3 \cup \LR{v}$}
            \State \textbf{else }\Return $\bot$
        \Else
            \State Let $x \in W$ be a leaf in $H[W]$
            \State $S_4 \gets \texttt{DisjointFVS}(H-x, \cM/x, W-x, N, t-1)$
            \State $S_5 \gets \texttt{DisjointFVS}(H, \cM, W-x, N\cup \LR{x}, t-1)$
            \State \textbf{if} $S_4 \neq \bot$ \textbf{ then } \Return{$S_4 \cup \LR{x}$}
            \State \textbf{else if} $S_5 \neq \bot$ \textbf{ then } \Return{$S_5$}
            \State \textbf{else} \Return $\bot$
        \EndIf
	\end{algorithmic}
\end{algorithm}

First, if there exists a vertex $v \in W$ such that $v$ has two special neighbors $u_1, u_2 \in N$ (by this, we mean that $u_1v, u_2v$ are special edges) such that $u_1, u_2$ belong to the same connected component in $H[N]$, then we make three recursive calls, say $1, 2$ and $3$: In the recursive call $i \in \LR{1, 2}$, the graph is obtained by deleting the special edge $u_iv$ and adding it to $\cE$. In the third recursive call, we add $v$ to the solution, contract the matroid $\cM$ on $v$, and remove $v$ from the graph. In each of the three recursive calls, the parameter $t$ drops by $1$. If a recursive call $1$ or $2$ returns a solution, then we return any of those. Otherwise, if call $3$ returns a solution $S$, then we return $S \cup \LR{v}$. If none of the calls returns a solution, then we return $\bot$.

Now assume that there is no such vertex $v \in W$ with two special neighbors $u_1, u_2 \in N$ from the same connected component of $H[N]$. In this case, we proceed with the same kind of branching strategy as for the standard FVS. We sketch it for completeness. Since $H[W]$ is a forest, it contains some leaf $x$. However, since each vertex has degree at least $3$ in $H$, it follows that $x$ must have at least two neighbors in $N$. In addition, these neighbors must be in different components. We make two recursive calls. In the first recursive call, we add $x$ to the solution, contract $\cM$ on $x$, and delete it from the graph. In this call, the parameter $t$ is reduced by $1$. If this call returns a solution $S$, then we return $S \cup \LR{x}$.
In the second recursive call, we add $x$ to $W$, and the parameter $t$ remains unchanged. If this call returns a solution, then we return the same solution. Otherwise, we return $\bot$.

At the bottom of the recursion, when $t = 0$, if the graph $G'$ is not acyclic, then we return $\bot$. Otherwise, we look at the special edges that have been added to $\cE$ over the course of the recursion. Recall that these edges correspond to the induced degree-2 paths in the original graph. Thus, we have a \yes-instance if and only if for each $e \in \cE$, we can select some $v_e \in P_e$, where $P_e$ is the set of vertices along the induced degree-2 path corresponding to $e$, such that $\bigcup_{e \in \cE} v_e \in \cI$. At this point, we observe that this can be reduced to a matroid intersection problem between the current matroid $\cM = (V(G'), \cI')$ (where $V(G') \subseteq V(G)$ -- here $G$ refers to the original kernelized graph) and a partition matroid $\cM_1 = (V(G), \cI_1)$, where $\cM_1$ is a matroid of rank $|\cE|$, such that a subset $S \subseteq V(G) $ is independent if and only if $|S \cap P_e| \le 1$. This problem can be solved in polynomial time even using oracle access to $\cM$ using, e.g., the algorithm of Cunningham~\cite{Cunningham86} (note that $\cM_1$ is a partition matroid, and thus oracle queries can be easily simulated in polynomial time). It is also straightforward to show that this leads to an algorithm of running time $9^t \cdot n^{\Oh(1)}$. Due to the iterative compression step, we obtain an overall running time of $10^\R \cdot n^{\Oh(1)}$. We conclude with the following theorem, by recalling that the solution size $\betak \le \R$.

\begin{restatable}{theorem}{fvstheorem}\label{thm:indfvs}
    {\sc Independent Feedback Vertex Set} admits a kernel of $\Oh(\betak^3)$ vertices and edges that can be computed in polynomial time in the oracle access model, where $\betak$ denotes the solution size, with $\betak \le \R = \rank(\cM)$. Further, the problem admits a deterministic algorithm running in time $10^\betak \cdot n^{\Oh(1)}$ in the oracle access model. 
\end{restatable}

\subsection{{\sc Independent OCT}} \label{subsec:indoct}

This is a standard application of the iterative compression technique. Suppose we have an OCT $S$ of size $k+1$ of a graph $G$, and we want to determine whether $G$ has an OCT $O$ that is independent in $\cM$. We try to guess $Y \coloneqq O \cap S$, assuming such an $O$ indeed exists. In this step, in addition to deleting $Y$ from the graph, we additionally contract the matroid $\cM$ on the set $Y$ of size $p \le k$, and obtain another matroid $\cM/Y$. Since $\cM$ is provided to us using its representation matrix, we can obtain that of $\cM_Y$ in polynomial-time (see, e.g., \cite{Marx09}). Thus, we can reduce the problem to {\sc Disjoint Independent Odd Cycle Transversal}, where we have an OCT $N \coloneqq S \setminus Y$ for the graph $G' = G-Y$, and we want to find an independent OCT of size $k-p$ in $G'$ that is disjoint from $N$. Then, by standard techniques (see, e.g., \cite{ReedSV04} or Lemma 4.14 of \cite{CyganFKLMPPS15}), this problem can be reduced to solving $2^{k-p}$ instances of {\sc Independent $(s, t)$-Cut} in $G'$, where the underlying matroid is now $\cM_Y$. Then, by using \Cref{thm:stcuttheorem}, we obtain the following theorem.

\begin{restatable}{theorem}{octthrorem} \label{thm:octtheorem}
{\sc Independent Odd Cycle Transversal} is \fpt parameterized by $k$, which denotes the rank of the given linear matroid. Specifically, it admits a $2^{\Oh(k^4 \log k)} \cdot n^{\Oh(1)}$ time  algorithm. 
\end{restatable}

\section{Vertex Multiway Cut}\label{sec:mwaycut}

	


Like \indstcut, our aim here is also to design an algorithm that solves a more general version of the problem, called \textsc{Generalized Independent Multiway Cut}, or \gindmcut~for short. 
By iteratively checking values $\betak = 1, 2, \ldots, $ and running the algorithm described below, we may inductively assume that we are working with a value $0 \le \betak \le \R$, such that there is no independent multiway cut for $T$ in $G$ for any $\betak' < \betak$. This follows from the fact that, our generalized problem, as defined below, requires the algorithm to return a set corresponding to $\betak'$ whose emptiness is equivalent (via \Cref{obs:repsetprops}, item 2) to concluding that the there is no independent multiway cut of size $\betak'$.

Then, let $\cF(T, \betak, G) \coloneqq \LR{S \subseteq V(G) \setminus T : |S| = \betak, S \in \cI, \text{ and $S$ is a multiway cut for $T$}}$. Note that, if there is no independent multiway cut of size $\betak$, then $\cF(T, \betak, G) = \emptyset$. On the other hand, if $\cF(T, \betak, G) \neq \emptyset$, then each $S \in \cF(T, \betak, G)$ must also be a \emph{minimal} multiway cut for $G$ -- otherwise, there would exist some $S' \subset S$ of size $\betak' < \betak$, and due to matroid properties $S' \in \cI$. Therefore, $\betak$ would not be the smallest such value. Next, we give a formal definition of the problem.

\begin{tcolorbox}[enhanced,title={\color{black} \textsc{Generalized Independent Multiway Cut}(~\gindmcut)}, colback=white, boxrule=0.4pt,
	attach boxed title to top center={xshift=-2cm, yshift*=-3mm},
	boxed title style={size=small,frame hidden,colback=white}]
	
	\textbf{Input:} An instance $(G, \cM, T, Q, \betak)$, where:
    \begin{itemize}[leftmargin=1.7cm]
        \item $G = (V, E)$ is a graph,
        \item $T \subseteq V(G)$ is a set of terminals, $Q \subseteq V(G)$ is a set of special vertices, with $T \subseteq Q$,
        \item Matroid $\cM = (U, \cI)$ of rank $r$, where $V(G) \setminus Q \subseteq U$,
        \item $0 \le \betak \le r$ 
    \end{itemize}
	\textbf{Output:} $\cF'(T, k, G) \repset{r-k} \cF(T, k, G)$.
\end{tcolorbox}

First, in \Cref{subsec:structural}, we analyze certain structural properties of solutions of \gimc, and use it to design an \fpt algorithm to find a so-called \emph{strong separator} for the instance, or to conclude that we have a \textsc{No}-instance. Then, in \Cref{subsec:recursive}, we use such a strong separator, and use to design a recursive algorithm for \gimc, and then show its correctness and that it is \fpt in $k$. 

\subsection{Finding a ``Strong Separator'' for \gimc} \label{subsec:structural}

In this subsection, we define the notion of a ``strong separator'' and how it interacts with any minimal multiway cut. Let $G$ be a graph and $T \subseteq V(G)$ be a set of terminals. Throughout the section, we will assume that $G$ is connected. Let $S \subseteq V(G)$ be a minimal multiway cut for $T$, and let $\cC_S$ denote the set of connected components of $G-S$. 
If a connected component $C \in \cC_S$ contains a terminal vertex then we call it a \emph{terminal component}; otherwise, we call it a \emph{non-terminal component}. Let $\cT_S$ denote the terminal components and $\cN_S$ denote the non-terminal components in $G-S$. Note that $\cC_S = \cT_S \uplus \cN_S$. We say that a $v \in S$ is adjacent to a component $C \in \cC_S$ if there exists some $u \in C$ such that $uv$ is an edge. 
For each $v \in S$, let $\cT(v, S) \subseteq \cT$ denote the set of terminal components that are adjacent to $v$. If the set $S$ is obvious from the context, then we denote $\cT(v, S)$ by $\cT_v$. Given two components $C_1$ and $C_2$ we say $C_1$ is adjacent to $C_2$ (and, vice versa) if there is a vertex in $C_1$ which has a neighbour in $C_2$. Because such a multiway cut $S$ (with certain additional properties, as stated below) will play a central role while designing the algorithm, we will often refer to it as a \emph{separator}, which is essentially a synonym for \emph{multiway cut for $T$}. We now define the notion of good and bad configuration corresponding to a minimal separator $S$.

\begin{definition}[Good and bad configurations]
	
	Let $G$ be a graph, $T \subseteq V(G)$ be a terminal set with $|T| \geq 3$, and $S$ be a minimal multiway cut $S$ for $T$. We refer to $\tup{G, T, S}$ as a \emph{configuration}. If there is a component $C$ in $G-S$ that is adjacent to each connected component in $G[S]$, then we say that $C$ is a \emph{large component}. We say $\tup{G, T, S}$ is a {\em bad configuration}, if (i) each component in $G[S]$ is adjacent to exactly two terminal components in $G-S$, and (ii) there is a large component \footnote{Observe that, since $|T| \ge 3$, in a bad configuration, there is exactly one large component.}. Otherwise, we say that $\tup{G, T, S}$ is a {\em good configuration}. 
\end{definition}

The definition of good and bad components easily lends itself to design a straightforward polynomial-time algorithm to check whether, corresponding to a given separator $S$, whether $\tup{G, T, S}$ is a good or a bad configuration. We state this in the following observation.

\begin{observation} \label{obs:goodconfigcheck}
    There exists a polynomial-time algorithm that takes input a graph $G$, and a set of terminals $T \subseteq V(G)$, and a minimal multiway cut $S \subseteq V(G) \setminus T$, and outputs whether $\tup{G, T, S}$ is a good or a bad configuration.
\end{observation}

In the following claims, we prove an important property of a good and a bad configuration, respectively.

\begin{figure}[ht!]
	\centering
	\begin{subfigure}{.5\textwidth}
		\centering
		\includegraphics[width=.7\linewidth]{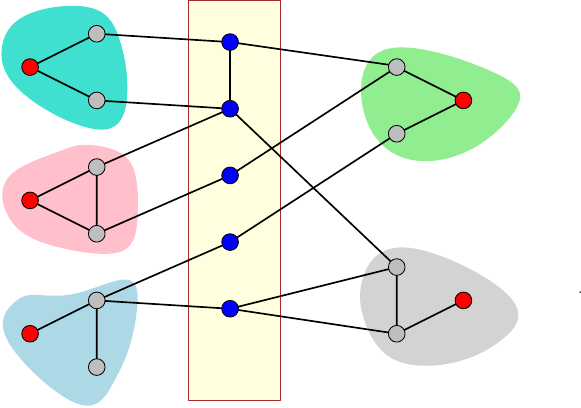}
		\label{fig:sub1}
	\end{subfigure}%
	\begin{subfigure}{.5\textwidth}
		\centering
		\includegraphics[width=.7\linewidth]{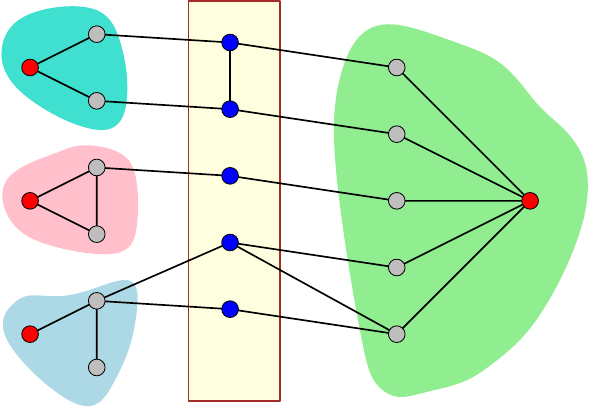}
		
		\label{fig:sub2}
	\end{subfigure}
	\caption{Examples of good (left) and bad (right) configurations. The blue and red vertices denote the separator and terminal vertices, respectively.}\label{fig:test1}
\end{figure}

\begin{claim}\label{cl:goodconfig}
	If $\tup{G,T,S}$ is a good configuration, then there is no multiway cut   that is fully contained in a terminal component in $G-S$. 
\end{claim}
\begin{claimproof}
	For a contradiction, assume that there is a multiway cut $S'$ and  a terminal component $C$ in $G-S$ such that $S' \subseteq V(C)$. We know that each component in $G[S]$ is adjacent to at least two terminal components.	
	As $\tup{G, T, S}$ is a good configuration, that means that either (i) there is a component $C_S$ in $G[S]$ such that $C_S$ is adjacent to at least three terminal components, or (ii) each component in $G[S]$ is adjacent to exactly two terminal components, but there is no large component. We analyze each case separately, and show that we arrive at a contradiction.
	
	In case (i), let $C_1, C_2, C_3$ be such three components in $G-S$ that are adjacent to the vertices of $C_S$. Among $C_1, C_2, C_3$, at most one can be $C$ -- hence at least two of them, say $C_1$, and $C_2$ must have empty intersection with the set $S'$. However, since $S' \subseteq C$, it follows that $S' \cap S =\emptyset$. This implies that, in $G-S'$, there is a path between the two terminals contained inside the components  $C_1$ and $C_2$, respectively. This is contradiction to the fact that $S'$ is a multiway cut for $T$. In case (ii), we know that each component in $G[S]$ is adjacent to exactly two terminal components, but there is no large component. Recall that $S' \subseteq V(C)$ for some terminal component $C$. By the case assumption $C$ is not a large component, which implies that there is a component $C'_S$ in $G[S]$ such that $C$ is not adjacent to $C'_S$. Also we know $C'_S$ is adjacent to at least two terminal components $C_1, C_2$, and none of them is $C$. As $S' \subseteq V(C)$ and $S \cap S' \neq \emptyset$, it follows that, in $G-S'$, there is path between two terminals contained in $C_1, C_2$, respectively, which again contradicts that $S'$ is a multiway cut for $T$. Hence, the claim follows.
\end{claimproof}

\begin{claim}\label{cl:badconfig}
	If $\tup{G, T, S}$ is a bad configuration then, except possibly the large component, no other terminal component contain a multiway cut.
\end{claim}

\begin{claimproof}
	Assume that $\tup{G, T, S}$ is a bad configuration with $C$ as large component. Suppose for the contradiction, assume that some $C_1$, distinct from $C$ is a terminal component that contains a multiway cut, say $S'$. Let $C_2$ be a terminal component (arbitrary choosen) other than $C$ and $C_1$ (note that such a component must exist since $|T| \ge 3$). As $C$ is adjacent to each component in $G[S]$, hence in $G-S'$, there is path between the two terminals that are contained in $C$ and $C_2$, respectively, which contradicts that $S'$ is a multiway cut. Hence, there is no multiway cut in $G-S$ that is fully contained in a terminal component that is not large.
\end{claimproof}

Next, we define the notion of \typea and \typeb scenarios corresponding to a minimal separator $S$.

\begin{definition}[\typea and \typeb scenarios] \label{def:scenarios}
	Given a graph $G$, a terminal set $T \subseteq V(G)$, and a minimal multiway cut $S$,  we say   the separator $S$ creates  a {\em \typea scenario} when no terminal component in $G-S $ can contain a multiway cut of size at most $|S|$. Otherwise, we say it creates a {\em \typeb scenario}.
\end{definition}

\begin{figure}[ht!]
	\centering
	\begin{subfigure}{.5\textwidth}
		\centering
		\includegraphics[width=.7\linewidth]{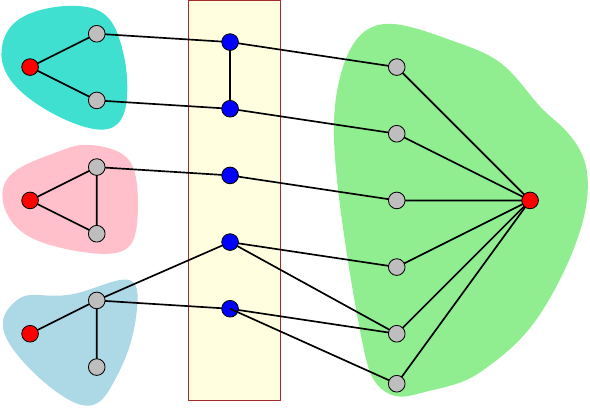}
		\label{fig:sub3}
	\end{subfigure}%
	\begin{subfigure}{.5\textwidth}
		\centering
		\includegraphics[width=.7\linewidth]{6.pdf}
		
		\label{fig:sub4}
	\end{subfigure}
	\caption{Example of a \typea (left) and \typeb (right) scenario. Blue and red vertices denote separator, and terminal vertices, respectively.}
	\label{fig:test2}
\end{figure}

\begin{claim}\label{obs:type}
	Given a graph $G$, a terminal set $T \subseteq V(G)$, and a minimal  multiway cut $S$ of size $\betak$, in $2^\betak \cdot  n^{\cO(1)}$ time we can check whether 	$S$ creates  an \typea scenario or not.
\end{claim}
\begin{claimproof}
	If $\tup{G, T, S}$ is a good configuration, then by \Cref{cl:goodconfig} we know that there is no multiway cut   that is fully contained in a terminal component in $G-S$. Hence, in this case $S$ creates \typea  scenario. Now, assume that $\tup{G, T, S}$ is a bad configuration. Now by \Cref{cl:badconfig} we see that except possibly the large component no other terminal component contains a multiway cut. Now to check whether large component contain a multiway cut we do the following. Let $C$ denote the large component. We take the subgraph $G_1 \coloneqq G[S \cup V(C)]$. Let $S_1, S_2, \ldots, S_q$ denote the partition of $S$ such that each $S_i$ corresponds to a connected component of $G[S]$. In the subgraph $G_1$, we contract each $S_i$, $1\le i \le q$ into a single vertex $s_i$. Let $G_2$ be the new graph after contraction. Let  $T' = \LR{s_1, \ldots, s_q, t}$. Note that $q \coloneqq |S'| \le |S| \le \betak$. Now we check whether there exists a multiway cut of size at most $\betak$ in the graph $G_2$ with terminal set $T'$. If \no then we conclude that   $S$ creates  a \typea  scenario in $G$ for the terminal set $T$. Otherwise, we have \typeb scenario. The correctness of the algorithm follows from the observation that $G$ has a multiway cut of size $\betak$ that is fully contained in $C$, if and only if $G_2$ has multiway cut of size at most $\betak$ with terminal set $T'$.
\end{claimproof}

Now we define a notion of a special type of multiway cut which we refer as strong separator which we use later in our algorithmic purpose.

\begin{definition}[Strong separator] \label{def:strongseparator}
	Given a configuration $\tup{G, T, S}$, we say $S$ is {\em strong separator} for $T$ in $G$ if (i) $\tup{G, T, S}$ is a good configuration, or (ii) $\tup{G, T, S}$ is a bad configuration, but $S$ creates a \typea scenario.	
\end{definition}

Now we prove an important result that leads to the design of a branching algorithm.

\begin{lemma}\label{lem:strongintersect}
	Let $(G, \cM, T, Q, \betak)$  be an instance of \gimc with $\betak\geq 2$, $|T| \ge 3$, and let $S \subseteq V$ be a strong separator of size at most $\betak$. There exists a polynomial-time algorithm to find a pair of terminal components in $G-S$ such that for any solution $O$, we either have $S \cap O \neq \emptyset $ or one of the components must contain at least one and at most $\betak-1$ vertices of $O$. 
\end{lemma}
\begin{proof}
	Assume that  $(G, \cM, T, Q, \betak, q)$  be an instance of \probmwcut  and    $S \subseteq V$ be a strong separator of size at most  $\betak$. 	 Due to the minimality of $S$, for each vertex $v \in S$ the set $S \setminus \{v\}$ is not a separator. Fix an arbitrary $v \in S$. It follows that, there exists a pair of terminals $t_1, t_2$, contained in terminal components $T_1, T_2$ respectively, such that we have  a path $P$ connecting $t_1$ and $t_2$ in $G$ such that the path intersects exactly once which is exactly at the vertex  $v$. Furthermore, $V(P) \subseteq V(T_1) \cup V(T_2) \cup \LR{v}$.
	Since $O$ is also a multiway cut for $T$, $O \cap V(P) \neq \emptyset$, which implies that $O \cap (V(T_1) \cup V(T_2) \cup \LR{v}) \neq \emptyset$. If $v \in O$ then we are done, as in this case $S \cap O \neq \emptyset$. Otherwise, $O \subseteq V(T_1 \cup T_2)$. However, since $S$ is strong separator, we know that $O \not\subseteq V(T_1)$, and $O \not\subseteq V(T_2)$. Hence, either $1 \leq  |V(T_1) \cap O| \leq \betak-1$ or $1 \leq  |V(T_2) \cap O| \leq \betak-1$. This completes the proof.
\end{proof}

\paragraph{Finding a Strong Separator}

In this section, we show that given a graph $G$ a terminal set $T$, an integer $\betak$ how to find a strong separator of size at most $t$ in $2^\betak \cdot n^{\cO(1)}$ time. First, we use a known \fpt algorithm of Cygan et al.~\cite{CyganPPW13}---that runs in time $2^\betak \cdot n^{\Oh(1)}$---to check whether $G$ has a minimal multiway cut of size at most $\betak$. Either this algorithm returns that there is no multiway cut of size $\betak$, or it returns a (w.l.o.g. minimal) multiway cut $S \subseteq V(G)$ of size at most $\betak$. In the former case, since there is no multiway cut of size $\betak$, there is no strong separator of size $\betak$ either -- hence we return that $G$ does not have a strong separator. Otherwise, in the latter case, we proceed to find a strong separator using the minimal separator $S$ returned by the algorithm, as described below. First, using \Cref{obs:goodconfigcheck}, we can check in polynomial-time whether $\tup{G, T, S}$ is a good or a bad configuration, and proceed to either of the following two cases.

\begin{description}
	\item[Case A. ($\tup{G,T, S}$ is a good configuration)] 
 
 In this case, we are done as per, that is, $S$ is a strong separator, as per our \cref{def:strongseparator}.
	
	\medskip 
	
	\item[Case B. ($\tup{G,T, S}$ is a bad configuration)] In this case, we can use \Cref{obs:type} to check in time $2^\betak \cdot n^{\Oh(1)}$ whether $S$ creates a type 1 or a type 2 scenario. In the former case, $S$ is a strong separator, and we are done. Thus, we are left with the case when $S$ creates a type 2 scenario, i.e., each component in $G[S]$ is adjacent to exactly two terminal components in $G-S$, and there is a large terminal component in $G-S$, say $C_L$. Assume that the component $C_L$ contains the terminal $t_L$. Note that via   If type 1 then we say that  $S$ is strong separator. Else 
    We try to find another separator of size $\betak$ which is fully contained in the large component using following procedure. Let $G_1 \coloneqq G-C_L$ and $G_2 \coloneqq G[S \cup C_L]$. Now consider the graph $G_2$.  For a terminal $t \in T \setminus \{t_L\}$, let $S_t \subseteq S$ denote the set of vertices in $S$ that are reachable from $t$ in the graph $G_1$. Now for each terminal vertex $t \in T \setminus \{t_L\}$ we identify all the vertices in $S_t$ to obtain a single merged vertex $v_t$. The merging operation preserves adjacencies towards the terminal components, i.e., there was a vertex $u\in V(G_2)$ such that $u$ has a neighbour in $S_t$ if and only if $u$ and $v_t$ are neighbors in the resulting graph.  Consider a  modified graph, say $H$,  with  vertex set  $\{v_t: t \in T \setminus \{t_L\}\} \cup V(C(t_L))$. We then find another minimal multiway cut of size $|S|$, with the new terminal set $T':= \{t_L\} \cup \{v_t: t \in T \setminus \{t_L\}\}$. The outcome is one of the following two possibilities: (i) either we find that there is no separator for $T'$ in $H$, or (ii) we find another separator $S^1$ for $T'$ in $H$. Again, we return to the beginning of this algorithm, by treating $S \gets S^1$, $T \gets T'$, and $H \gets G$, and check whether the new $\tup{G, T, S}$, i.e., $\tup{H, T', S^1}$ is a good and bad configuration, and so on.
    Suppose that in a particular iteration, let $n_T$ denote the number of terminal vertices in current graph. Then in next iteration, the  size of the vertex set strictly decreases, more specifically it decreases by at least $(n_T-1)$. Hence, this algorithm runs for at most $n$ iterations, and eventually we find a vertex set $S'$. Below we show that the set $S'$ is indeed a strong separator for the original set of terminals $T$ in the original graph $G$. 
\end{description}

\begin{claim}\label{claim:strong}
    Let $S'$ be the vertex set  returned by the aforementioned procedure. Then $S'$ must be a strong separator for $T$ in $G$.
\end{claim}

\begin{claimproof}
    First, we show that $S'$ is the separator for $T$ in $G$. Assume that our procedure runs for $\ell$ iterations, and let $S^0 = S, S^1, \ldots, S^\ell = S'$ be the separators found in each iteration.
    We show by induction that $S'$ is a separator for $T$. In the base case, trivially $S^0 = S$ is a separator for $T$ in $G$. Assume that for some $0 \le i \le \ell-1$, $S^i$ is the separator in $G$. Clearly $S^i$ is not a strong separator, otherwise we have returned $S^i$ in our procedure. As $S^i$ is not a strong separator, then by definition $\tup{G,T,S^i}$ is a bad configuration and $S^i$ creates a type 2 scenario. Now in our process, while working with the separator $S^i$, we consider a modified graph, say $H$, with a set of vertices $\{v_t: t \in T \setminus \{t_L\}\} \cup V(C(t_L))$ and find another minimal multiway cut $S^{i+1}$ with the set of terminals $T':= \{t_L\} \cup \{v_t: t \in T \setminus \{t_L\}\}$. For the sake of contradiction, assume that $S^{i+1}$ is not a separator for $T$ in $G$, which means that there is a pair of terminal vertices $s$ and $t$ such that there is a path $P_{st}$ between $s$ and $t$ in $G- S^{i+1}$. First, we consider the case where none of the $s$ and $t$ vertices is $t_L$. As $\tup{G,T,S^i}$ is a bad configuration, the path must consist of a pair of vertices of $S^{i+1}_s$ and $S^{i+1}_t$. Hence we got a subpath connecting $v_t$ and $v_s$ of the path $P_{st}$ in $H - S^{i+1}$, which is a contradiction to the fact that $S^{i+1}$ is a separator for the terminal set $\{v_t: t \in T \setminus \{t_L\}\}$ in $H$. The same argument holds when $t_L \in \{s,t\}$. Therefore, by induction, we have shown that $S^{\ell}$ is a separator for $T$ in $G$. Now, it remains to show that $S^{\ell}$ is a \emph{strong} separator for $T$ in $G$. Suppose not, then $\tup{G,T,S^{\ell}}$ must be a bad configuration, and $S^\ell$ must create a type 2 scenario. And our process must be able to find another separator $S''$ that is a contradiction to our termination condition. 
\end{claimproof}

\begin{lemma}\label{lem:strong}
	There exists an algorithm that takes an input an instance $(G, \cM, T, Q,  \betak)$ of \gimc, where $G$ is connected, $\betak\geq 2$, and $|T| \geq 3$,  and in time $2^\betak \cdot n^{\Oh(1)}$, either returns that $\cF(T, \betak, G) = \emptyset$, or returns a strong separator $S$ corresponding to the instance, and two components $C_s, C_t$ satisfying the property stated in \Cref{lem:strongintersect}. Here, $n = |V(G)|$.
\end{lemma}

From the above discussion, we obtain the following corollary.

\begin{corollary} \label{cor:strong}
	An instance $(G, \cM, T,Q, \betak)$ of \gimc, where $G$ is connected, $\betak \ge 2$, and $|T| \ge 3$, if $\cF(T, \betak, G) \neq \emptyset$ then $G$ has a strong separator of size at most $\betak$. 
\end{corollary}

\subsection{A Recursive Algorithm for \gindmcut} \label{subsec:recursive}

We design a recursive algorithm called \mwcut for the problem \textsc{Generalized Independent Multiway Cut} (\gindmcut) that takes the following as input.
\begin{itemize}
	\item A graph $G = (V, E)$, 
    \item a set $T \subseteq V(G)$ of terminals, 
    
    \item a set  $Q \subseteq V(G)$  of special vertices, with $T \subseteq Q$,
        \item Matroid $\cM = (U, \cI)$ of rank $r$, where $V(G) \setminus Q \subseteq U$, and 
	\item two non-negative integers $\betak$, $q$ where $\betak+q \leq r$. 
\end{itemize}
We use ${\cal J} = (G,  \cM, T, Q, \betak, q)$ to denote the input instance to the algorithm. Given this instance $\cJ$, our goal of algorithm \mwcut is to find $\cF'(T, \betak, G) \repset{q} \cF(T, \betak, G)$ of bounded size. For the values $\betak \gets r$ and $q \gets 0$, the solution of the instance
$(G, \cM, T, Q, r, 0)$ corresponding to the representative set of the entire solution of size $r$ for \probmwcut in $G$ for the terminal set $T$.

\begin{description}[leftmargin=7pt]

	\item[Step 0: Preprocessing \label{step:prepocess}] We perform the following preprocessing steps exhaustively in the order in which they are mentioned. The soundness of these rules is immediate.
	\begin{description}[leftmargin=3pt]
		\item[Rule 0.a.\label{step:disconnterm1}] If there exists a non-terminal vertex $u \in V(G)$ that has no path to any terminal vertex in $T$, then make a recursive call  \mwcut$(G - u, \cM - u, T , Q, \betak,q)$. If the recursive call returns \textsc{No}, then we return 
		\no to original instance. Otherwise, return the solution returned by the recursive call.

		\item[Rule 0.b.\label{step:disconnterm2}] If there exists a terminal vertex $t \in T$ that is disconnected from all other terminals in $T\setminus \LR{t}$ in $G$, then make a recursive call \mwcut$(G - t, \cM - t, T - t, Q,\betak,q)$. If the recursive call returns \textsc{No}, then we return 
		\no to original instance. Otherwise, return the solution returned by the recursive call.

  \item[Rule 0.c.\label{step:edgeterminal}] If there exists a pair of adjacent terminal vertices, then return \textsc{No}.

        The following observation is immediate.
        \begin{obs} \label{obs:preproc}
            After an exhaustive application of Rule \hyperref[step:disconnterm1]{0.a} and \hyperref[step:disconnterm2]{0.b}, each connected component of $G$ contains at least two terminals of $T$. 
        \end{obs}
        
        \item[Rule 0.c.\label{step:termnumbers}] At this point we have exhaustively applied the previous two rules. Now, if the number of connected components (each of which contains at least two terminals by \Cref{obs:preproc}) of $G$ is more than $\betak$, then we return \no. 

	\end{description}

		\item[Step 1: Base Case.\label{step:base}] Here we have following two cases. 
        \\Case a: If $\betak=1$, then $\cF(G, T, 1)$ consists of singleton vertex-sets $\LR{v}$, such that $\LR{v} \in \cI$, and $\LR{v}$ is a multiway cut for $T$ (of size $1$). For each $v \in V(G) \setminus T$, we can easily check whether $v \in \cF(G, T, 1)$ in polynomial time; this implies that the set $\cF(G, T, 1)$ can be found in polynomial time. Then, we compute $\cF'(G, T, 1) \repset{q} \cF(G, T, 1)$ in 
        $n^{\cO(1)}$ time, and return $\cF'(G, T, 1)$. The correctness of this step is immediate.
        \\Case b: If $|T|=2$, then we solve the instance $(G, \cM, Q, s, t, \betak, q)$ of \indstcut using \Cref{thm:stcuttheorem}, and return the family returned by the algorithm. The correctness of this step follows from that of \Cref{thm:stcuttheorem}, and the theorem implies that this step runs in time $f(r, k) \cdot n^{\Oh(1)}$ and returns a family $\cF'(G, \{s,t\}, \betak)$ of size $2^r$. The correctness of this step is immediate.

        \begin{claim} \label{cl:basecase}
		Let $\cF'(G, T, \betak)$ denote the set that we return in the base case. Then $\cF'(G, T, \betak) \repset{q}  \cF(T, \betak, G)$.
	\end{claim}
	
	\medskip
	
	\item[Step 2: Solving the problem in a disconnected graph. \label{step:connected}] In this step, we describe how to deal with the problem when the underlying graph is not connected. Let $(G, \cM, T, Q, \betak,q)$ be the instance and $G$ be the union of $\ell$ connected components $G_1, \ldots G_\ell$. Observe that after an exhaustive application of the preprocessing (\hyperref[step:prepocess]{Rule 0.b}), each component in $G$ has at least two terminal components, so each minimal solution must have a non-empty intersection with each terminal component. Hence, we can assume that $\ell \leq \betak$. Let $T_i \coloneqq T \cap V(G_i)$ and $Q_i \coloneqq Q \cap V(G_i)$. 	Now for each component $G_i$ we try to find the smallest $\betak_i \leq \betak$ such that $G_i$ has an independent multiway cut of size $\betak_i$ -- which can be found by making a recursive call to the algorithm for each component separately, starting from $\betak'_i =1, 2, \ldots, \betak-1$, in order to find the smallest value $\betak_i$ such that the recursive call for all $\betak'_i < \betak_i$ returned \no, and $\betak_i$ returns a non-empty family. Note that, assuming the number of connected components is more than $1$, if for some $i \in [\ell]$, we find that for $\betak'_i = \betak-1$, the instance $(G_i, \cM, T_i, Q_i, \betak_i, q+\betak-\beta_i)$ is a \no-instance, then we conclude that we have a \no-instance.
    The correctness of this follows from the fact that any $S \in \cF(T, \betak, G)$ must have a non-empty intersection with each connected component. Hence, in each recursive call, the parameter strictly decreases. Note that we make a recursive call to most $\betak^2$ many sub-problems. Let $\betak_i$ be the smallest value for each  component $G_i$ such that we find non-empty solution for $(G_i, \cM, T_i,Q_i, \betak_i,q+\betak-\betak_i)$. Note that the $\betak_i$'s form a natural partition $\cP$ of $\betak$ such that $\betak = \betak_1 + \ldots + \betak_\ell$. Let $\cF'_i$ be the output of the $i$th call with parameter $\betak_i$. We compute $\cF'_1 \bullet \cF'_2 \bullet \cdots \bullet \cF'_\ell$, and we filter out any sets from the resultant that do not constitute a minimal independent multiway cut of size $\betak$ for $T$. Let the remaining collection of sets be $\cF^1(T, \betak, G) $. Finally, we compute and return $\cF'(T, \betak, G) \repset{q} \cF^1(T, \betak, G)$. 
	
	\begin{claim} \label{cl:disconnected}
		$\cF'(T, \betak, G) \repset{q} \cF(T, \betak, G)$.
	\end{claim}
	\begin{claimproof}
    We show the claim in two parts. First, we show that $\cF'(T, \betak, G)\subseteq \cF(T, \betak, G)$, and then we show that the former also $q$-represents the latter. For the first part, consider an $R \in \cF'(T, \betak, G)$. Hence, for each $i \in [\ell]$, there exists some $R_i \in \cF'(T_i, \betak_i, G_i) \repset{q+\betak-\betak_i} \cF(T_i, \betak_i, G_i)$, such that, (i) each $R_i$ is a minimal independent multiway cut of size $\betak_i$ for $T_i$ in $G_i$, (ii) $R = R_1 \cup R_2 \cup \ldots \cup R_\ell$. It follows that $R$ is a minimal multiway cut of size $\betak$ for $T$ in $G$, and thus $R \in \cF(T, \betak, G)$.

	Now consider $Y \subseteq U(\cM)$ be a set of size $q$ such that there is a set $O  \in \cF(T, \betak,  G)$  (i.e., $O$ is a minimal multiway cut of size $\betak$ for $T$), such that $O $ fits $Y$.  We have to show that there is a $\hat{X}  \in \cF'(T, \betak, G)$, such that $\hat{X}$ fits $Y$, that is also a minimal multiway cut of size $\betak$ for $T$. 
	
	We know that $\cF'_i$ is the output of the recursive call on $(G_i, \cM, T_i, Q_i, \betak_i, q+\betak-\betak_i)$. By induction (recall that $\betak_i < \betak$), we know that $\cF'_i \repset{q+\betak-\betak_i} \cF(G_i, \betak_i, T_i)$. 
    Let $O_i \coloneqq O \cap V(G_i)$, $Y_i \coloneqq Y \cap V(G_i)$, and $Y_{out} = Y \setminus V(G)$. Note that $Y = Y_1 \cup Y_2 \cup \ldots \cup Y_\ell \cup Y_{out}$. Now we provide a recursive definition of $\hat{X_i}$ for each $i \in [\ell]$. The set $\hat{X_1}$ is a representative for the set $O_1$ in the graph $G_1$ for the $\betak+q-\betak_1$ sized set $(O \cup Y) \setminus O_1$. Now for $i>1$ 
    $$Z_i \coloneqq \lr{\bigcup_{j<i} (\hat{X_j} \cup Y_j)} \cup Y_i \cup \lr{\bigcup_{j>i} (O_j \cup Y_j)} \cup Y_{out}.$$
	
	 Clearly, $|Z_i|= \betak+q-\betak_i$. By the property of $\cF'_i$, there exists some $\hat{X_i} \in \cF'_i$, such that $\hat{X_i} $ is a minimum multiway cut of size $\betak_i$ for $T_i$ in $G_i$, and $\hat{X_i}$ fits  $Z_i$. Now, let $\hat{X} \coloneqq \hat{X_1} \cup \cdots \cup \hat{X_{\ell}}$. Clearly $|\hat{X}|= \sum_{i=1}^{\ell} \betak_i= \betak$.  As $\hat{X_i}$ is a solution for $G_i$ and $G$ consists of $\ell$ components $G_i, i\in [\ell]$,  so $ \hat{X} \coloneqq \hat{X_1} \cup \cdots \cup \hat{X_{\ell}} $ is a solution for $G$. By our recursive definition, we can conclude that $\hat{X_{\ell}} \cup Z_{\ell} $ is an independent set. Since $\hat{X_{\ell}} \cup Z_{\ell} = \hat{X} \cup Y$, we are done.
	\end{claimproof}

%
%
%
%
%
%
%
	
	\medskip
	
	\item[Step 3: Finding a Strong Separator.\label{step:normalsol}] Since \hyperref[step:connected]{Step 2} is not applicable, we know that the graph $G$ is connected. Then we use our \fpt algorithm of \Cref{lem:strong} to check whether $G$ has a strong separator corresponding to the instance $(G, \cM, T, Q, \betak,q)$ that runs in time $2^\betak \cdot n^{\Oh(1)}$. Either this algorithm returns a (w.l.o.g. minimal) separator $S \subseteq V(G)$ of size at most $\betak$, or it returns \textsc{No}. In the former case, we proceed to the next step, i.e., \hyperref[step:guessintersection]{Step 4}. In the later case, we return \no to our original instance $(G, \cM, T, Q, \betak,q)$. Because if $G$ does not admit a multiway cut for $T$ of size $\betak$, then it also does not admit an independent vertex multiway cut of size $\betak$. This shows the soundness of returning \no.
	
	\medskip

	\item[Step 4: Branching procedure. \label{step:guessintersection}] Let $\cF \coloneqq \cF(G, T, \betak)$ denote the family of feasible solutions for the given instance $\cJ = (G, \cM, T, Q, \betak, q)$. If $\cF \neq \emptyset$, then \Cref{lem:strongintersect} implies that, there exists a pair of components $C_s$ and $C_t$, such that for any $O \in \cF$, for each strong separator $S$ of size $\betak$ corresponding to the terminal set $T$, we  have that, either (i) $S \cap O \neq \emptyset $, or (ii) either  $1 \leq  |C_s \cap O| \leq \betak-1$, or $1 \leq  |C_t \cap O| \leq \betak-1$. Furthermore, as stated in \Cref{lem:strongintersect}, given a strong separator $S$, such a pair of components can be found in polynomial-time. 
    However, since a-priori we do not know whether case (A) or (B) is applicable, we design a branching strategy that considers both possibilities. We describe this branching in the following. 
	
	\begin{description}
		\item[Handling case (A): Branching on the vertices in $S$:] We branch on each vertex $v$ on $S$ assuming that $v \in S \cap O$.   More specifically, we do the following: For each vertex $v \in  S$ we make a recursive call to \mwcut$(G-v, \cM, T, Q \setminus \{v\}, \betak-1, q)$. Let $\cF'(v)$ be the set returned by the recursive call. Then $\cF^1(v) \coloneqq \cF'_v \bullet \LR{v}$.

			\begin{claim} \label{cl:s-branchingrepset}
				$\cF^1(v) \repset{q} \cF(T, \betak, v, G)$, where $\cF(T, \betak, v, G) \subseteq \cF(T, \betak, G)$ is the set of all independent minimal multiway cuts for $T$ that contain $v$. 
			\end{claim}
			
	\begin{claimproof}
        First we argue that $\cF^1(v) \subseteq \cF(T, \betak, v, G)$, and then we show that $\cF^1(v)$ also $q$-represents $\cF(T, \betak, v, G)$. Consider some $R \in \cF^1(v)$, and $R$ must be equal to $Q \cup \LR{v}$ for some $Q \in \cF'_v$. By induction hypothesis, $Q \in \cF(G-v, T, \betak-1)$, i.e., $Q$ is a miminmal multiway cut for $T$ of size $\betak-1$ in $G-v$. It follows that $R = Q \cup \LR{v}$ is also a minimal multiway cut of size $\betak$ for $T$ in $G$. Since $v \in R$, this shows that $\cF^1(v) \subseteq \cF(T, \betak, v, G)$.

        Now, consider some $Y \subseteq U(\cM)$ of size $q$ for which there is a set $O  \in \cF(T, \betak, v, G)$  such that $Y$ fits $O$. First, by the induction hypothesis, we have $\cF'(v) \repset{q+1} \cF(T, \betak-1,  G-v) $. This implies, in particular, that each $X \in \cF'(v)$, is a minimal multiway cut of size $\betak-1$ for $T$ in $G-v$.     
        We have to show that there is a $\hat{X}  \in \cF^1(v)$, such that $\hat{X}$ is a minimal multiway cut of size $\betak$ in  such that $\hat{X}$ fits $Y$. 	As $\cF(T, \betak, v, G) \subseteq \cF(T, \betak, G)$ is the set of all independent minimal multiway cuts for $T$ that contain $v$, the vertex $v$ must belong to $O$. Consider the set $O_v \coloneqq O \setminus \{v\}$ and $Y_v \coloneqq Y \cup \{v\}$. Clearly $O_v \cup Y_v = O \cup Y$, and for any $X \in \cF'(v)$, $v \not\in X$.  
        So for the set $Y_v \subseteq U(\cM)$,  of size $q+1$  there is a set $O_v  \in \cF(T, \betak-1,  G-v)$, such that $O_v$ fits $Y_v$. Hence, there exists $X_v \in \cF'(v)$ such that $X_v$ also fits $O_v$. Further, $X_v \in \cF(T, \betak-1, G-v)$ as well. Now, consider the set $\hat{X} \coloneqq X_v \cup \{v\}$. Obviously, $\hat{X}$ is disjoint from $Y$. As $X_v$ is an independent multiway cut for $T$ in $G-v$, so $\hat{X} = X_v \cup \LR{v}$ is an independent multiway cut in $G$. Finally, since $\hat{X} \cup Y= X_v \cup Y_v \in \cI$, the claim follows.
	\end{claimproof}
		
		\smallskip 
		
		\item[Handling case (B): Branching on terminal components in $G-S$:] For a terminal vertex $x \in T$, we use the notation $C_x$ to denote the terminal component in $G-S$ containing the vertex $x$. Recall the components $C_s$ and $C_t$ returned by \Cref{lem:strongintersect} with the property that, if $\cJ$ is a \yes-instance, and for some $O \in \cF(T, \betak, G)$, it holds that $O \cap S = \emptyset$, then either $1 \le |O \cap V(C_s)| \le \betak-1$, or $1 \le |O \cap V(C_t)| \le \betak-1$.
        First we guess (i.e., branch) which components among $C_t$ and $C_s$ have the above mentioned property. WLOG, assume that the component $C_t$ has such property, and the other case is symmetric. 	
        
        Let $\reach(t)$ denote the set of vertices in $S$ which is reachable (connected) from the vertex $t$ in the induced subgraph $G[C_t \cup S]$.
		Next, we guess the following two things (i) the value $|V(C_t) \cap O|$ and (ii) the connectivity among the vertices in $S \cup \{t\}$ in the subgraph $G[V(C_t)  \cup \reach(t)]- O$.  More specifically,  we  guess a partition of the vertices in $\reach(t) \cup \{t\}$ such that it holds the following property: a pair of vertices $x,y \in \reach(t) \cup \{t\}$ belongs to the same set of the partition if and only if there is a path between $x$ and $y$ in $G[V(C_t)  \cup \reach(t)]- O$. We introduce the following notation.
        \begin{definition}
            Given a graph $H$ and partition $\Pi$ of vertex set $U \subseteq V(H)$, the graph obtained in the following procedure is denoted by $H^*/\Pi$: for each part $\pi \in \Pi$, we \identity all the vertices of $\pi$ into a single vertex, say $v_{\pi}$ (we refer preliminary section to detail about the \identity operation). 
        \end{definition}
        
        Now in our algorithm,  for the  component $C_t$,  for each partition  $\cP$ of the vertex set $ \reach(t) \cup \LR{t}$, for each integer $1 \le \betak_1 \le \betak-1$ (let $\betak_2 = \betak-\betak_1$),
        we create two instances of \probmwcut and recursively call \mwcut on each of them. We highlight that the parameters for both the instances satisfy $1 \le \betak_1, \betak_2 \le \betak-1$, hence the parameter strictly decreases in both the recursive calls. These two instances are defined as follows. 
        
		  \begin{figure}[ht!]
	\centering
	\includegraphics[scale=0.5]{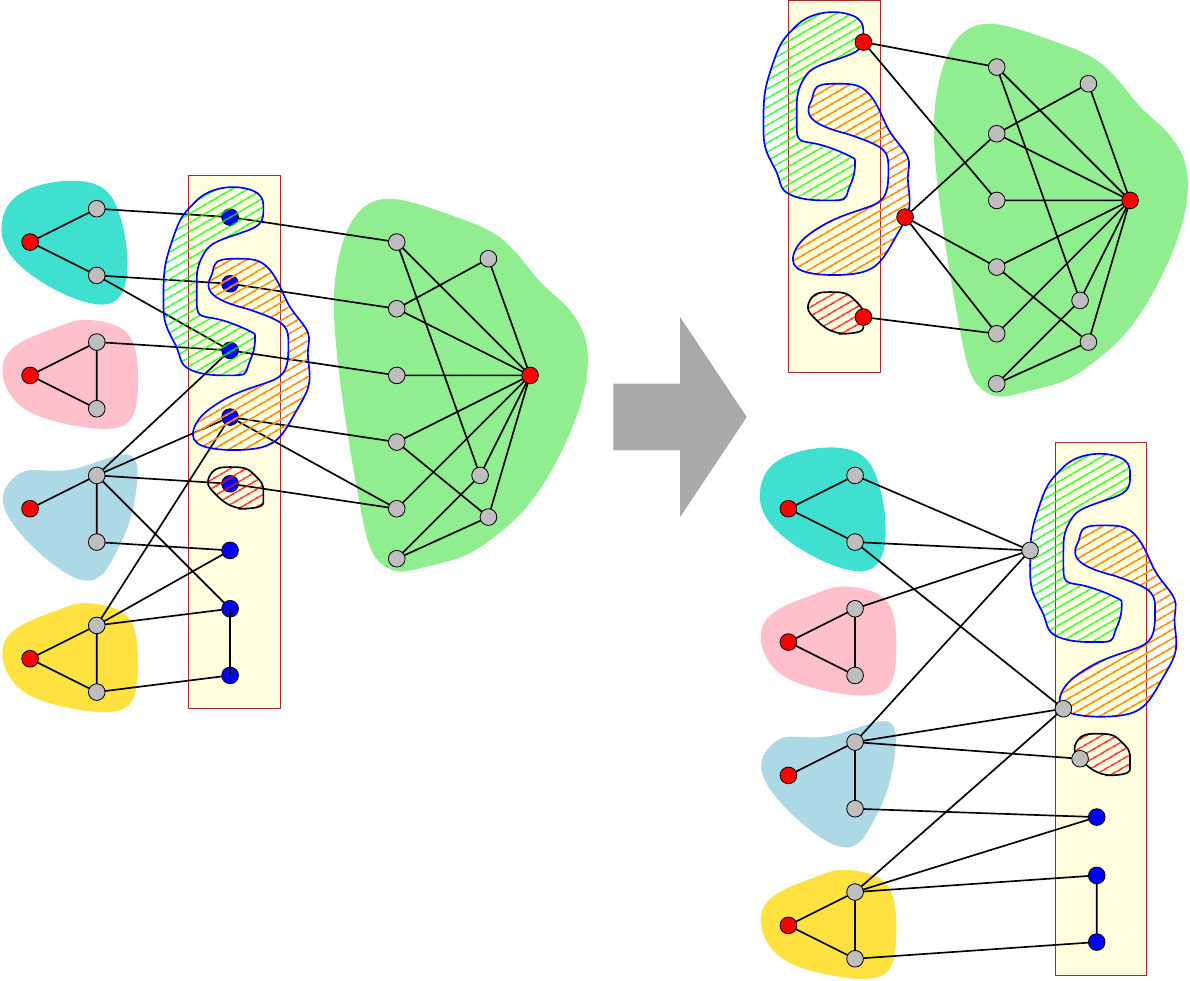}
	\caption{An illustration of  recurrsive algorithm when branching on terminal components.}
	\label{fig:branching}
\end{figure}
		
		\begin{itemize}[leftmargin=2pt]
			\item Instance 1 is $\cJ_1 \coloneqq (G_1, \cM, T_1, Q_1, \betak_1, q+\betak_2)$. Here $G_1$ is $H^*/\cP$, where $H:= G[V(C_t)  \cup \reach(t)]$. That is  $V(G_1) = V(C_t) \cup V(\cP)$ where $V(\cP)$ is the set of all new vertices added after processing \identity operation on each $P \in \cP$. And the terminal set $T_1:= V(\cP)$, $Q_1:= T_1$. Let $\cF'(T_1, \betak_1, G_1) \repset{q+\betak_2} \cF(T_1, \betak_1, G_1)$ is the output returned by recursive call. 
            By induction this is true. 
			
			\item Instance 2 is $\cJ_2 \coloneqq (G_2, \cM, T_2, Q_2, \betak_2, q+\betak_1)$. Here, the description of the graph $G_2$ and terminal set $T_2$ is based on the partition $\cP$.  We have the following two cases depending on whether the vertex $t$ is identified with some other vertices in $\reach(t)$ in $\cP$ or not. Let $P_t$ denotes the set in the partition in $\cP$ which contains $t$.
			
			\begin{description}[leftmargin=3pt]
				\item[Case 4.a. $|P_t|=1$  \label{case:t=0}] Here $t$ is not in same part in $\cP$ with any $v \in \reach(t)$. Note that this implies that we require that $t$ is disconnected from all $v \in \reach(t)$ in each solution of $\cJ_1$. Since any path between $t$ and $t'$ in $G$ must use a vertex of $\reach(t)$, it follows that $G - X$ already disconnects $t$ from all $t' \in T\setminus \LR{t}$, for any solution $X$ of $\cJ_1$. Hence, in $\cJ_2$, we may remove $t$ from $T$.
                \\More formally, we define $G_2 \coloneqq G^*/\cP- V(C_t)$, $T_2 \coloneqq T \setminus \{t\}$, and $Q_2 \coloneqq T_2 \cup \big(V(\cP) \setminus \LR{t}\big)$. That is,  $V(G_2) = \big(V(G) \setminus \big(V(C_t) \cup \reach(t))\big) \cup \big(V(\cP) \setminus \LR{t}\big)$. 
				
				\smallskip

				\item[Case 4.b. $|P_t|\geq 2$  \label{case:t=1}] In this case, the terminal $t$ is identified with at least one vertex from $\reach(t)$. Let $v_t$ denote the result of the \identity operation on the set $P_t$.
                \\In this case $G_2 \coloneqq G^*/\cP- (V(C_t)\setminus \{t\})$, $T_2 \coloneqq T \setminus \{t\} \cup \{v_t\}$, and $Q_2 \coloneqq T_2 \cup V(\cP) $. That is  	$V(G_2) = \lr{V(G) \setminus \lr{V(C_t) \cup \reach(t)}} \cup V(\cP)$. 
                \\Let $\cF'(T_2, \betak_2, G_2) \repset{q+\betak_1} \cF(T_2, \betak_2, G_2)$ is the output returned by recursive call. By induction this is true.

			\end{description}
		More formally, let ${\cal G}_t$ as the collection of all tuples $(C_t, \cP, \betak_1)$ of guesses corresponding to component $C_t$ -- analogously, we also define the set of guesses ${\cal G}_s$ corresponding to $C_s$, and also recursively call the algorithm for each guess in ${\cal G}_s$, the pair sub-instances obtained in a similar way (we omit this description).
  
 		Now to return the solution for branching on terminal components, we do the following. First we compute $ \cF'(T_1, \betak_1, G_1) \bullet \cF'(T_2, \betak_2, G_2)$, and from the resultant, we filter out any sets that do not form an independent minimal multiway cut for $T$ in $G$. Let us refer to the remaining $\betak$-family as $\cF^1(C_t, \cP, \betak_1)$. 
    We say a minimal independent multiway cut $O$ is {\em compatible} with the guess $(C_t, \cP, \betak_1)$ if following two conditions hold: (i) $|O \cap V(C_t)|=\betak_1$, (ii) for any pair of vertices $u,v \in \reach(t) \cup \{t\}$, if $u$ and $v$ belong to  two different sets in the partition $\cP$, then there is no path from $u$ to $v$ in the subgraph $G[V(C_t) \cup \reach(t) \setminus O]$.
				
				\begin{claim} \label{cl:divide}
					$\cF^1(C_t, \cP, \betak_1) \repset{q} \cF(T, \betak, G, C_t, \cP, \betak_1)$, where $\cF(T, \betak, G, C_t, \cP, \betak_1) \subseteq \cF(T, \betak, G)$ is the set of minimal independent multiway cuts that are compatible with the guess $(C_t, \cP, \betak_1)$. 
				\end{claim}
			
			\begin{proof}
            First, we show that $\cF^1(C_t, \cP, \betak_1) \subseteq \cF(T, \betak, G, C_t, \cP, \betak_1)$. To this end, consider any $R \in \cF^1(C_t, \cP, \betak_1)$, which implies that there exists some $P \in \cF'(T_1, \betak_1, G_1), Q \in \cF'(T_2, \betak_2, G_2)$, such that $R = P \cup Q$, $P \cap Q = \emptyset$, and $R$ is indeed a minimal multiway cut of size $\betak$ for $T$ in $G$. Firstly, note $R \cap V(C_t) = P$, which implies that $|R \cap V(C_t)| = |P| = \betak_1$. Secondly, since $P \in \cF'(T_1, \betak_1, G_1) \repset{q+\betak_2} \cF(T_1, \betak_1, G_1)$, it follows that, for any $u, v \in \reach(t) \cup \LR{t}$, such that $u, v$ belong to different parts of $\cP$, there is no path from $u$ to $v$ in $G_1 - P$. Therefore, there is no path from $u$ to $v$ in $G[(V(C_t) \cup \reach(t))] - P = G[(V(C_t) \cup \reach(t)) \setminus R]$. This shows that $R$ is indeed compatible with the guess $(C_t, \cP, \betak_1)$, i.e., $R \in \cF(T, \betak, G, C_t, \cP, \betak_1)$. 
            

            Next, we show that $\cF^1(C_t, \cP, \betak_1)$ also $q$-represents $\cF(T, \betak, G, C_t, \cP, \betak_1)$. To this end, consider $Y \subseteq U(\cM)$ of size $q$, for which there exists some $O  \in \cF(T, \betak, G, C_1, \cP, \betak_1)$ such that $O$ fits $Y$. We show that there exists some $\hat{X} \in \cF^1(C_t, \cP, \betak_1)$ such that $\hat{X}$ also fits $Y$.    
            Let $O_i \coloneqq O \cap V(G_i)$, $Y_i \coloneqq Y \cap V(G_i)$, and $Y_{out} \coloneqq Y \setminus V(G)$. Note that $Y = Y_1 \cup Y_2 \cup Y_{out}$. Now we define the set $\hat{X_i}$ for each $i \in [2]$. The set $\hat{X_1}$ is a representative for the set $O_1$ in the graph $G_1$ for the $q+\betak_2$ sized set $(O \cup Y) \setminus O_1$. Let  $X_2= \hat{X_1} \cup Y$. As $|X_2|= q+\betak_1$ by induction, there exists some $\hat{X_2} \in \cF'(T_2, \betak_2, G_2)$, such that $\hat{X_2} $ is a minimum multiway cut of size $\betak_2$ for $T_2$ in $G_2$, and $\hat{X_2}$ fits  $O_2$.
				 Now, let $\hat{X} \coloneqq \hat{X_1} \cup  \hat{X_2}$. Clearly $|\hat{X}|= |\hat{X_1}| +  |\hat{X_2}|= \betak_1 + \betak_2$.  By our argument, we can conclude that $\hat{X_{2}} \cup X_{2} $ is an independent set. Since $\hat{X_{2}} \cup X_{2} =\hat{X_{2}} \cup \hat{X_1} \cup Y = \hat{X} \cup Y$, we have that $\hat{X} \cup Y$ is an independent set. 
     
                It remains to show that $ \hat{X} $ is a minimal multiway cut for $T$ in $G$. Consider a shortest path $P_{12}$ between two terminal vertices, say $t_1, t_2 \in T$ in $G-\hat{X}$.

                \begin{description}
                    \item[Case 1] First, we analyze the case where $t \notin \{t_1, t_2\}$. Naturally, $P_{12}$ is also a shortest path between $t_1$ and $t_2$. In this case, the path $P_{12}$  must use a vertex from $V(C_t)$, otherwise $\hat{X_2}$ is not a solution for $\cJ_1$. We first claim that, $|V(P_{12}) \cap \reach(t)| \ge 2$. Suppose to the contrary that $|V(P_{12} \cap \reach(t)| \le 1$, then there exists a path $P'_{12}$ in $G_2 - \hat{X}_2$ between $t_1$ and $t_2$, which contradicts the fact that $\hat{X}_2$ is a solution for $\cJ_2$. Thus, $|V(P_{12}) \cap \reach(t)| \ge 2$, Now consider distinct $x_1, x_2 \in V(P_{12}) \cap \reach(t)$, such that the internal vertices of the subpath $P'_{12}$, say $I$, between $x_1$ and $x_2$ satisfies that $\emptyset \neq I \subseteq V(C_t)$. 
                    Now, suppose that
                    $x_1$ and $x_2$ are adjacent in $G$, then we can obtain a path between $t_1$ and $t_2$  by replacing $x_1, V(I), x_2$ by the edge $x_1x_2$ in $P_{12}$, which yields a shorter path between $t_1$ and $t_2$, which is a contradiction to the assumption that $P_{12}$ is a shortest path. In summary, we have a subpath $P'_{12}$ of $P_{12}$ such that the endpoints of $P'_{12}$, that is, $x_1, x_2 \in \reach(t)$, and $\emptyset\neq I \subseteq V(C_t)$. Now we consider different cases and arrive at a contradiction.
                \\Case (1a): 
                 $x_1$ and $x_2$ are in the same part $P \in \cP$. Then $P$ is contracted into a single vertex $v_P$ in $G_2$. This implies that there is also a path between $t_1$ and $t_2$ in $G_2-\hat{X}_2$, which is a contradiction.
                 \\Case (1b): $x_1$ and $x_2$ are in different parts $P_1$ and $P_2$ respectively, which are identified into the vertices $v_1, v_2$ in $G_1$, respectively. However, this implies that there is a path $P'_{12}$ between $v_1$ and $v_2$ in $G_1-\hat{X}_1$, which contradicts that $X_1$ is a solution for $\cJ_1$.

                \item[Case 2] Now, we analyze the case where $t \in \{t_1, t_2\}$. WLOG, assume that $t=t_1$; therefore, $P_{12}$ is the shortest path between $t$ and $t_2$ in $G-\hat{X}$.        
                In this case, as $t \in V(C_t)$ the path $P_{12}$  must use a vertex from $(V(C_t) \setminus \{t\}) \cup \reach(t)$.  In path $P_{12}$, let $x_2$ denote the vertex in $\reach(t)$ such that all the internal vertices of the subpath connecting $x_2$ to $t$ along $P_{2}$ are from $V(C_t)$. Now we consider different cases and arrive at a contradiction.  \\Case (2a): 
                 $x_2$ and $t$ are in the same part $P \in \cP$.  Then it follows that there is a path from $t_2$ to $v_P$ in $G - \hat{X_2}$, which contradicts the assumption that $ \hat{X_2} $ is a solution for $\cJ_2$.
                            \\Case (2b): $t$ and $x_2$ are in different parts. Now, we note that there is a path between $t$ and $x_2$ in             $G[V(C_t) \cup \reach(t)] - \hat{X_1}$, it contradicts the assumption that $ \hat{X_1} $ is a solution for $\cJ_1$.
                                
                \end{description}

                This completes the claim. 
			\end{proof}

   Here is a simple observation.

   \begin{claim}\label{claim:compatible}
       Consider any $O \in \cF(T, \betak, G)$ such that $O \cap S = \emptyset$. Then, there is some guess $(C_t, \cP, \betak_1)$ with which $O$ is compatible.
       
   \end{claim}

   \begin{claimproof}
  Assuming that $O \cap S = \emptyset$, we know that at least one of the components $C_s$ or $C_t$ has a nonempty intersection with the vertices of $O$. Now we obtain a guess that is compatible with $O$. Assume that $|V(C_t) \cap O|=\betak_1 > 0$.  Let $\cP$ be the partition of the vertices in $\reach(t) \cup \{t\}$ defined as follows. A pair of vertices $u$ and $v$ from $\reach(t) \cup \{t\}$ belongs to a different part if they are not connected in the graph $G[V[C_t] \cup \reach(t)] - O$ and the graph $G- (V(C_t) \cup O)$. It is easy to see that guess $(C_t, \cP, \betak_1)$ is compatible with $O$.
   \end{claimproof}

		\end{itemize}

	\end{description}

	\medskip 
	
	\item[Step 5: Return as output] We return \textsc{no} if each of the branching steps returns \no to their corresponding instance. Otherwise, we do the following.
	Notice that in each branching step, we have at most $\betak+ 2(\betak-1)d_{\betak+1} = \betak^{\Oh(\betak)}$ many recursive calls, where $d_{\betak+1}$ denotes the $(\betak+1)$ th Bell's number: number of  possible partitions of a set of $(\betak+1)$ elements. First we take union of all the solution sets corresponding to every recursive calls, specifically, we compute 

    \begin{equation}
        \cF^1(T, \betak, G) \coloneqq \bigcup_{v \in S} \cF^1_v \cup  \bigcup_{(C_t, \cP, \betak_1) \in {\cal G}_t} \cF^1(C_t, \cP, \betak_1)  \cup \bigcup_{(C_s, \cP, \betak_1) \in {\cal G}_s} \cF^1(C_t, \cP, \betak_1). \label{eqn:repsetcomputation1}
    \end{equation}
    Then we compute 
     \begin{equation}
         \cF'(T, \betak, G) \repset{q} \cF^1(T, \betak, G) \label{eqn:repsetcomputation2}
     \end{equation}
    using \Cref{prop:repset} and return the set $\cF'(T, \betak, G)$ as the output for the instance $\cJ = (G,\cM, T, Q, \betak, q)$.  
	

	
	
\end{description}
This finishes the description of \mwcut.

\subparagraph{Correctness.}

	\begin{lemma}[Safeness of \mwcut]
		Our branching algorithm returns the set  $\cF'(T, \betak, G) \repset{q} \cF(T, \betak, G)$ of size $2^{\betak+q}$.  
	\end{lemma}
	\begin{proof}
		Recall that $\cF(T, \betak, G)$ is the set of all minimal independent multiway cuts for $T$ of size $\betak$ in $G$. First, we consider base case, then the correctness follows from \Cref{cl:basecase}.         Now, suppose that neither of the base cases is applicable. If the graph $G$ is disconnected, then again the correctness of the algorithm follows directly from \Cref{cl:disconnected}.        Hence, finally let us consider the case when $G$ is connected, $\betak \ge 2$, and $|T| \ge 3$.  
        First, by \Cref{lem:strong}, we either correctly conclude that $\cJ = (G, \cM, T,Q, \betak, q)$ is a \no-instance, or we find a strong separator $S$ and a pair of components $C_s, C_t$ in time $2^\betak \cdot n^{\Oh(1)}$, such that for any minimal solution $O \in \cF(G, T, \betak)$ of size $\betak$, either (i) $S \cap O = \emptyset$, (ii) $1 \le |C_s \cap O| \le \betak-1$, or (iii) $1 \le |C_t \cap O| \le \betak-1$. Further, any $O$ satisfying (ii) (resp.~(iii)) must be compatible  with some guess $(C_t, \cP, \betak_1) \in {\cal G}_t$ (resp.~$(C_s, \cP, \betak_1) \in {\cal G}_s$) (by \cref{claim:compatible}). 
        This implies that, 
        \begin{align}
            \cF(T, \betak, G) = \bigcup_{v \in S} \cF(T, \betak, v, G) \ \cup  &\bigcup_{(C_t, \cP, \betak_1) \in {\cal G}_t} \cF(T, \betak, G, C_t, \cP, \betak_1)  \nonumber
            \\&\cup \bigcup_{(C_s, \cP, \betak_1) \in {\cal G}_s} \cF(T, \betak, G, C_s, \cP, \betak_1).
        \end{align}
        Then, by \Cref{cl:s-branchingrepset,cl:divide}, it follows that $\cF^1(T, \betak, G) \repset{q} \cF(T, \betak, G)$, where the former is defined in \eqref{eqn:repsetcomputation1}. Then, since $\cF'(T, \betak, G) \repset{q} \cF^1(T, \betak, G)$ as computed in \eqref{eqn:repsetcomputation2}, by \Cref{obs:repsetprops}, it follows that $\cF'(T, \betak, G) \repset{q} \cF(T, \betak, G)$. Finally, by \Cref{prop:repset}, it follows that $|\cF'(T, \betak, G)| \le 2^{\betak+q}$.        
		\end{proof}
	
	\subparagraph{Time analysis.}
	
	Note that in each of the recursive calls to our branching steps, the parameter $\betak$ decreases by at least $1$. The number of partition of the vertices in $S$ is bounded by $\betak^\betak$. So the degree at each node of the branching tree is bounded by $\betak^{\cO(\betak)}$. In the base case where $|T|=2$ or $\betak=1$ we found the corresponding solution in time $f(r, k) \cdot n^{\Oh(1)}$. The number of nodes in the tree is bounded by $\betak^{\cO(\betak)}$. Now we calculate the running to find a solution to each node assuming that we have all the solution to its child node. For this, we have to take union of at most $(\betak+ 2(\betak-1)(\betak+1)^{\betak+1})^\betak$ many sets where the size of each set is at most $\binom{\betak+q}{\betak}$. This union takes $\binom{\betak+q}{\betak} \cdot (\betak+ 2(\betak-1)(\betak+1)^{\betak+1})^\betak$ time. Then finding the $q$-representative set will take time $\binom{k+q}{k}$ using at most $\Oh(|{\cal A}|(\binom{k+q}{k})k^{\omega}+(\binom{k+q}{k})^{\omega-1})$ (by \cref{prop:repset}, where $\cal A$ is the size of the set after taking the union.

	\medskip 
	
	Hence, we have the following theorem.

 \begin{restatable}{theorem}{multiwaycuttheorem} \label{thm:multiwaycuttheorem}
	There exists a deterministic algorithm for {\sc Generalized} \probmwcut that runs in time $f(r, k) \cdot n^{\Oh(1)}$, and outputs $\cF'(T, \betak,G) \repset{q} \cF(T, \betak, G)$ such that $|\cF'(T, \betak,G)| = 2^{\betak+q}$. Here $\cF(T, \betak, G)$ denotes the set of all $\betak$ size multiway cuts for $T$ in $G$, assuming that there is no independent multiway cut of size $\betak' < \betak$. 
\end{restatable}
	
	

\section{Lower bound}\label{sec:lower}
In this section, we show an unconditional lower bound for the \textsc{Independent $(s,t)$-Cut} problem when the input matroid is given by the independence oracle. The proof is inspired by the reductions used for the similar statements in~\cite[Theorem~1]{fomin2024stability}.

\begin{theorem}\label{thm:lb-uncond}
There is no algorithm solving \textsc{Independent $(s,t)$-Cut} with matroids of rank $k$ represented by the independence oracles using $f(k)\cdot n^{o(k)}$ oracle calls for any computable function $f$.
\end{theorem}

\begin{proof}
Let $p$ and $q$ be positive integers. We construct the graph $G_{p,q}$ as follows (see~\Cref{fig:Gpq}).
\begin{itemize}
\item Construct two vertices $s$ and $t$.
\item For each $i\in[p]$, construct a set $A_i=\{a_1^i,\ldots,a_q^i\}$ of $q$ vertices and a set $B_i=\{b_1^i,\ldots,b_q^i\}$ of $q$ vertices.
\item For each $i\in[p]$, make $s$ adjacent to $a_1^i$ and $b_1^i$, and make $t$ adjacent to $a_q^i$ and $b_q^i$.
\item For each $i\in[p]$ and $j\in[q-1]$, make $a_j^i$ adjacent to $a_{j+1}^i$ and $b_{j+1}^i$, and make $b_j^i$ adjacent to $a_{j+1}^i$ and $b_{j+1}^i$.
\end{itemize}
We have that $G_{pq}$ has $2pq+2$ vertices. Notice that each minimum $(s,t)$-cut in $G_{p,q}$ has size $2p$ and is of form $\bigcup_{i=1}^p\{a_{j_i}^i,b_{j_i}^i\}$ for some $j_1,\ldots,j_p\in[q]$.

\begin{figure}[t]
\centering
\scalebox{0.75}{
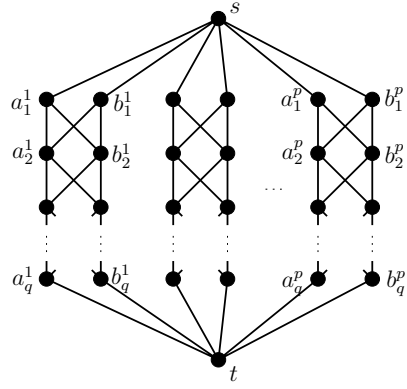}
\caption{Construction of $G_{p,q}$.}
\label{fig:Gpq}
\end{figure}

Consider a family of indices $j_1,\ldots,j_p\in[q]$ and set 
$W=\bigcup_{i=1}^p\{a_{j_i}^i,b_{j_i}^i\}$.
We define the matroid $\M_W$ with the ground set $V(G_{p,q})$ where each independent set is a subset of $V(G_{p,q})\setminus\{s,t\}$
as follows for $k=2p$:
\begin{itemize}
\item  Each set $X\subseteq V(G_{p,q})\setminus\{s,t\}$ of size at most $k-1$ is independent and any set of size at least $k+1$ is not independent.
\item A set $X\subseteq V(G_{p,q})\setminus\{s,t\}$ of size $k$ is independent if and only if either $X=W$ or there is $i\in\{1,\ldots,p\}$ such that
$|A_i\cap X|\geq 2$ or  $|B_i\cap X|\geq 2$ or there are distinct $h,j\in[q]$ such that 
$a_{h}^i,b_{j}^i\in X$.
\end{itemize}
We denote by $\I_W$ the constructed family of independent sets.  
In fact, $\M_W$ is the same as the matroid constructed in the proof of 
\cite[Theorem~1]{fomin2024stability}. We provide the proof of the following claim showing that $\M_W$ is a matroid for completeness even if it identical to the proof of \cite[Claim~1.1]{fomin2024stability} in the proof of Theorem~1 therein.

\begin{claim}[{\cite[Claim~1.1]{fomin2024stability}}]\label{cl:W-matr}
$\M_W=(V(G_{p,q}),\I_W)$ is a matroid of rank $k=2p$.
\end{claim}

\begin{proof}[Proof of \Cref{cl:W-matr}]
We have to verify that $\I_W$ satisfies the independence axioms (1)--(3). The axioms (1) and (2) for $\I_W$ follow directly from the definition of $\I_W$. To establish (3), consider arbitrary $X,Y\in \I_W$ such that $|X|<|Y|$. If $|X|<k-1$ then for any $v\in Y\setminus X$, $Z=X\cup \{v\}\in \cI_W$ because $|Z|\leq k-1$.

Suppose $|X|=k-1$ and $|Y|=k$. 
If  there is $i\in[p]$ such that $|A_i\cap X|\geq 2$ or  $|B_i\cap X|\geq 2$ or there are distinct $h,j\in[q]$ such that 
$a_{h}^i,b_{j}^i\in X$ then for any $v\in Y\setminus X$, the set $Z=X\cup \{v\}$ has the same property and, therefore, $Z\in\cI_W$.
Assume that this is not the case. By the construction of $G_{p,q}$, we have that for each $i\in[p]$, $|X\cap A_i|\leq 1$ and $|X\cap B_i|\leq 1$, and, furthermore, 
there is $j\in[q]$ such that $X\cap (A_i\cup B_i)\subseteq \{a_{j}^i,b_{j}^i\}$. 
Because $|X|=k-1$, we can assume without loss of generality that there are indices $h_1,\ldots,h_p\in [q]$ such that $X\cap (A_i\cup B_i)= \{a_{h_i}^i,b_{h_i}^i\}$ for $i\in[p]$ and $X\cap (A_p\cup B_p)=\{a_{h_p}^p\}$.
Recall that $W=\bigcup_{i=1}^p\{a_{j_i}^i,b_{j_i}^i\}$ for $j_1,\ldots,j_p\in[q]$. If there is $v\in Y\setminus X$ such that $v\neq b_{j_p}^p$ then consider $Z=X\cup\{v\}$.  We have that there is $i\in[p]$ such that
$|A_i\cap Z|\geq 2$ or  $|B_i\cap Z|\geq 2$ or there are distinct $h,j\in[q]$ such that 
$a_{h}^i,b_{j}^i\in Z$, that is,   $Z\in \I_W$. 
Now we assume that $Y\setminus X=\{b_{j_p}^p\}$. Then $Y=W$ and we can take $v=b_{j_p}^p$. We obtain that $X\cup\{v\}=Y\in \I_W$. 

By the construction of $\I$, the rank of $\M$ is $k$. This concludes the proof.
\end{proof}

We show the following lower bound for the number of oracle queries for frameworks $(G_{p,q},\cM_W)$; the proof is similar to the proof of \cite[Claim~1.2]{fomin2024stability} and we provide it for completeness.

\begin{claim}\label{cl:W-lower}
Solving \textsc{Independent $(s,t)$-Cut} for instances $(G_{p,q},\M_W,s,t)$ with the matroids $\cM_W$ defined by the independence oracle for an (unknown) stable set $W$ of $G_{p,q}$ of size $k=2p$ demands at least $q^p-1$ oracle queries.   
\end{claim}

\begin{proof}[Proof of \Cref{cl:W-lower}] 
Recall that any $(s,t)$-cut in $G_{p,q}$ has size at least $k$ and all sets $X\subseteq V(G_{p,q})\setminus \{s,t\}$ of size at least $k+1$ are not independent with respect to $\M$.  Thus, we are looking for an $(s,t)$-cut of size $k$. For every $(s,t)$-cut $X$ in $G_{p,q}$ of size $k$,  $X=\bigcup_{i=1}^p\{a_{h_i}^i,b_{h_i}^i\}$ for some $h_1,\ldots,h_p\in[q]$. By the definition of $\M$, any $X$ of this form is independent if and only if $X=W$. This means that the problem boils down to identifying unknown $W$ using oracle queries. 
Querying the oracle for sets $X$ of size at most $k-1$ or at least $k+1$ does not provide any information about $W$. Also, querying the oracle for $X$ of size $k$ with the property that there is $i\in[p]$ such that
$|A_i\cap X|\geq 2$ or  $|B_i\cap X|\geq 2$ or there are distinct $h,j\in[q]$ such that 
$a_{h}^i,b_{j}^i\in X$  also does not give any information because all these are independent. 
Hence, we can assume that the oracle is queried only for sets $X$ of size $k$ with the property that for each $i\in[p]$, there is $j\in[q]$ such that $X\cap V(G_i)=\{a_{j}^i,b_{j}^i\}$, that, is the oracle is queried for 
$(s,k)$-cuts of size  $k$. The graph $G_{p,q}$ has $q^p$ such sets. Suppose that the oracle is queried for at most $q^p-2$ stable sets of size $k$ with the answer {\sf no}. Then there are two distinct stable sets $W$ and $W'$ of size $k$ such that the oracle was queried neither for $W$ nor $W'$. The previous queries do not help to distinguish between $W$ and $W'$. Hence, at least one more query is needed. This proves the claim. 
\end{proof}

To show the theorem, suppose that there is an algorithm $\mathcal{A}$ solving  \textsc{Independent $(s,t)$-Cut} with at most $f(k)\cdot n^{g(k)}$ oracle calls on matroids of rank $k$ for computable functions $f$ and $g$ such that $g(k)=o(k)$. Without loss of generality, we assume that $f$ and $g$ are monotone non-decreasing functions.
Because $g(k)=o(k)$, there is a positive integer $K$ such that $g(k)<k/2$ for all $k\geq K$. Then for 
each $k\geq K$, there is a positive integer $N_k$ such that for every $n\geq N_k$, $(f(k)\cdot n^{g(k)}+1)k^{k/2}<(n-2)^{k/2}$. 

Consider instances $(G_{p,q},\cM_W)$ for even $k\geq K$ where $p=k/2$ and $q\geq (N_k-2)/k$. We have that $k=2p$ and $n=2pq+2$. Then $\mathcal{A}$ applied to such instances would use at most $f(k)\cdot n^{g(k)}<\big(\frac{n-2}{k}\big)^{k/2}-1=q^p-1$ oracle queries contradicting \Cref{cl:W-lower}. This completes the proof. 
\end{proof}

\section{Conclusion} \label{Sec:conclusion}
In this paper we studied matroidal generalizations  of several fundamental graph separation problems, such as {\sc Vertex $(s, t)$-Cut} and {\sc Vertex Multiway Cut}.  We showed that these problems are \fpt when parameterized by the size of the solution. These results were obtained by combining several recent and old techniques in the world of parameterized complexity. Our paper leaves several interesting questions open. These include obtaining  $2^{\Oh(k)} \cdot n^{\Oh(1)}$ time algorithm for all the problems studied in this paper. In addition, designing \fpt algorithms for matroid versions of other problems such as {\sc Vertex Multi-cut}, {\sc Directed Feedback Vertex Set} remains an interesting research direction.

    \bibliography{main}

\end{document}